\documentclass[a4paper,11pt,oneside,color,abstract=true,titlepage=false]{scrartcl}

\usepackage{microtype}
\usepackage[english]{babel}

\usepackage[pdftex]{graphicx,color}
\usepackage{amsmath,amssymb}

\usepackage[amsmath,thmmarks,amsthm]{ntheorem}

\usepackage[linesnumbered]{algorithm}%[plain,]
\usepackage[noend]{algorithmic}

\usepackage{enumerate}

\usepackage{tabularx}

\theoremseparator{.}
\theoremstyle{plain}
\newtheorem{theorem}{Theorem}[section]
\newtheorem{lemma}[theorem]{Lemma}
\newtheorem{myclaim}[theorem]{Claim}

\newtheorem{proposition}[theorem]{Proposition}
\theoremstyle{nonumberplain}
\renewtheorem{theorem*}{Theorem}
\renewtheorem{lemma*}{Lemma}
\renewtheorem{corollary*}{Corollary}
\renewtheorem{proposition*}{Proposition}
\theorembodyfont{\upshape}
\theoremstyle{change}
\newtheorem{definition}[theorem]{Definition}
\theoremstyle{nonumberplain}
\renewtheorem{definition*}{Definition}
\theoremstyle{change}
\theoremheaderfont{\bfseries}

\theoremstyle{nonumberplain}
\renewtheorem{remark*}{Remark}
\renewtheorem{example*}{Example}
\usepackage{xspace}

\newcommand{\ie}{i.\nolinebreak[4]\hspace{0.125em}e.\nolinebreak[4]\@\xspace}

\usepackage{mathptmx}
\DeclareMathAlphabet{\mathcal}{OMS}{cmsy}{m}{n}

\usepackage[bold,basic,small]{complexity}
\usepackage{mathtools} %for \DeclarePairedDelimiter
\usepackage{mleftright} %for \mleft\{ ... \;\middle|\; ... \mright\}
\usepackage{pdfpages} %\includepdf
\usepackage{tikz}
\usetikzlibrary{arrows,automata,calc,decorations.markings,positioning,shapes.geometric}
\usepackage{stmaryrd}
\usepackage{caption}
\usepackage{subcaption} %\begin{subfigure}
\usepackage{textcomp} % \textcelsius
\usepackage{authblk}
\definecolor{darkgreen}{RGB}{0,180,0}

\makeatletter
\tikzstyle{initial by arrow}=   [after node path=
{
	{
		[to path=
		{
			[->,double=none,every initial by arrow]
			([shift=(\tikz@initial@angle:\tikz@initial@distance)]\tikztostart.\tikz@initial@angle)
			coordinate [anchor=\tikz@initial@anchor] {\tikz@initial@text}
			-- (\tikztostart)}]
		edge ()
	}
}]
\makeatother

\tikzset{negated/.style={
		decoration={markings,
			mark= at position 0.5 with {
				\node[transform shape] (tempnode) {$\bigg/$};
			}
		},
		postaction={decorate}
	}
}

\DeclarePairedDelimiter\abso{\lvert}{\rvert}
\newcommand{\abs}{\abso}

\newclass{\ALOGSPACE}{ALOGSPACE}
\newclass{\EXPTIME}{EXPTIME}
\newclass{\NEXPTIME}{NEXPTIME}
\newcommand{\dotCup}{\mathbin{\dot{\cup}}}
\newcommand{\call}[1]{\widehat{#1}}
\newcommand{\hist}[1]{\overline{#1}}
\newcommand{\hs}{\hist{\Sigma}^*}

\newcommand\algproblem[3]{%
	\par\noindent \textsc{#1}\\
	{\textit{Given}}: #2\\
	{\textit{Question}}: #3\par
}
\newcommand{\Call}{\ensuremath{\mathit{Call}}\xspace}
\newcommand{\Read}{\ensuremath{\mathit{Read}}\xspace}
\newcommand{\Move}{\mathit{Move}}
\newcommand{\Inf}{\mathit{Inf}}
\newcommand{\Next}{\mathit{Next}}
\newcommand\restr[2]{{% we make the whole thing an ordinary symbol
		\left.\kern-\nulldelimiterspace % automatically resize the bar with \right
		#1 % the function
		\vphantom{\big|} % pretend it's a little taller at normal size
		\right|_{#2} % this is the delimiter
}}
\newcommand{\states}[3][]{\mathit{states}_{#1}(#2,#3)}
\newcommand{\term}{\emph}

\newcommand{\tFive}{\ensuremath{\SReg}}

\usepackage{hyperref} % MUSS LETZTES Usepackage sein!!!!!
\hypersetup{pdfborder = {0 0 0}} % keine roten Boxen um Hyperrefs

\newcommand{\calA}{\ensuremath{{\mathcal A}}\xspace}

\newcommand{\calD}{\ensuremath{{\mathcal D}}\xspace}

\newcommand{\calK}{\ensuremath{{\mathcal K}}\xspace}
\newcommand{\calN}{\ensuremath{{\mathcal N}}\xspace}
\newcommand{\calP}{\ensuremath{{\mathcal P}}\xspace}

\newcommand {\SReg}      {\ensuremath{\textsc{SReg}}\xspace}

\newcommand{\df}{\ensuremath{\mathrel{\smash{\stackrel{\scriptscriptstyle{
						\text{def}}}{=}}}} \;}

\newcommand{\shat}{\ensuremath{\hat{\sigma}}}

\newcommand{\sse}{\ensuremath{\textsf{sse}}}
\newcommand{\lesl}{\ensuremath{\le_{\textsf{sl}}}}
\newcommand{\lsl}{\ensuremath{<_{\textsf{sl}}}}
\newcommand{\gsl}{\ensuremath{>_{\textsf{sl}}}}
\newcommand{\online}[1][\calN]{\ensuremath{\textsc{OnlineNFA}(#1)}}

\newenvironment{proofof}[1]{\begin{proof}[Proof of #1.]}{\end{proof}}

\newcommand{\short}[1]{}\newcommand{\full}[1]{#1}  % Langversion

\usetikzlibrary{backgrounds,fit}

\usetikzlibrary{arrows,decorations.pathmorphing,backgrounds,positioning,fit,petri,backgrounds}
\usetikzlibrary{calc,3d}
\usetikzlibrary{scopes}
\usetikzlibrary{patterns}
\usepgflibrary{shapes.geometric,shapes.arrows}
\usepgflibrary{shapes,decorations.pathreplacing}
\usetikzlibrary{matrix}

\title{Streaming Rewriting Games: Winning Strategies and Complexity}

\author[1]{Christian Coester\footnote{Supported by EPSRC}}
\author[2]{Thomas Schwentick}
\author[3]{Martin Schuster}

\affil[1]{University of Oxford}
\affil[2]{TU Dortmund}
\affil[3]{University of Edinburgh}

\begin{document}

\maketitle

\hyphenation{NFAUniversality UFAEmptiness}

\begin{abstract}
Context-free games on strings are two-player rewriting games based on a set of production rules and a regular target language. In each round, the first player selects a position of the current string; then the second player replaces the symbol at that position according to one of the production rules. The first player wins as soon as the current string belongs to the target language. In this paper the one-pass setting for context-free games is studied, where the knowledge of the first player is incomplete, she selects positions in a left-to-right fashion and only sees the current symbol and the symbols from previous rounds. The paper studies conditions under which dominant or undominated strategies for the first player exist and investigates the complexity of some related algorithmic problems.
\end{abstract}

\pagenumbering{arabic}
\section{Introduction}

\emph{Context-free games} on strings are rewriting games based on a set of \emph{production rules} and a regular \term{target language}. They are played by two players, \emph{Juliet} and \emph{Romeo}, and consist of several rounds. In each round, first Juliet selects a position of the current string; then Romeo replaces the symbol at that position according to one of the production rules. Juliet wins as soon as the current string belongs to the target language. Context-free games were introduced by Muscholl, Schwentick and Segoufin~\cite{MuschollSS06} as an abstraction of \emph{Active XML}.

Active XML (AXML) is a framework that extends XML by ``active nodes''. In AXML documents, some of the data is given explicitly while other parts are given by means of embedded calls to web services \cite{MiloAABN05}. These embedded calls can be invoked to materialise more data. As an example (adapted from \cite{MiloAABN05,MuschollSS06}; see Figure~\ref{fig:AXML}), consider a document for the web page of a local newspaper. The document may contain some explicit data, such as the name of the city, whereas information about the weather and local events is given by means of calls to a weather forecast service and an events service (see Figure~\ref{fig:AXMLbefore}). By invoking these calls, the data is materialised, \ie replaced by concrete weather and events data (see Figure~\ref{fig:AXMLafter}). The data returned by the service call may contain further service calls.

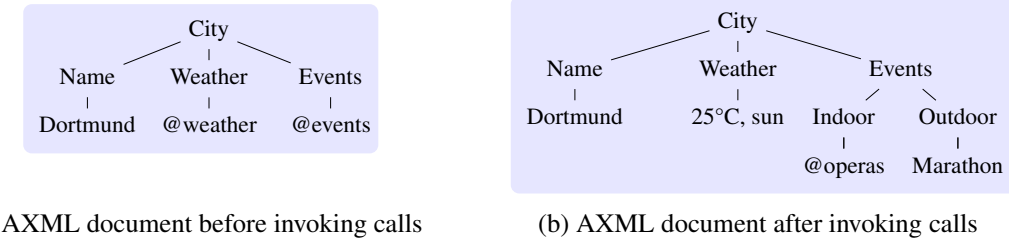
\begin{figure}
\subcaptionbox{AXML document before invoking calls\label{fig:AXMLbefore}}[.46\textwidth]{
\scalebox{0.8}{
%\pgfdeclarelayer{background}
	\begin{tikzpicture}[yscale=0.8,xscale=1.0,every node/.style={font=\vphantom{Ag}}]%font=\vphantom{Ag} bewirkt, dass Text je Zeile nach derselben Baseline ausgerichtet, egal wie hoch die Buchstaben sind
	\node (Stadt) at (2,2) {City};
	\node (Name) at (0,1) {Name};
	\node (Dortmund) at (0,0) {Dortmund};
	\node (Wetter) at (2,1) {Weather};
	\node (weather) at (2,0) {@weather};
	\node (Veranstaltungen) at (4,1) {Events};
	\node (events) at (4,0) {@events};
	\node (Kino) at (3.3,-1) {\phantom{@operas}};%dummy for space
	
	\draw[-] (Stadt) -- (Name);
	\draw[-] (Stadt) -- (Wetter);
	\draw[-] (Stadt) -- (Veranstaltungen);
	\draw[-] (Name) -- (Dortmund);
	\draw[-] (Wetter) -- (weather);
	\draw[-] (Veranstaltungen) -- (events);
%	\boundingbox;
 \begin{pgfonlayer}{background}
    \node [fill=blue!10,rounded corners,fit=(Stadt) (Dortmund) (Veranstaltungen)] {};
  \end{pgfonlayer}
\end{tikzpicture}}%
}
\subcaptionbox{AXML document after invoking calls\label{fig:AXMLafter}}[.54\textwidth]{
\scalebox{0.8}{
	\begin{tikzpicture}[yscale=0.8,xscale=1.34,every node/.style={font=\vphantom{Ag}}]%font=\vphantom{Ag} bewirkt, dass Text je Zeile nach derselben Baseline ausgerichtet, egal wie hoch die Buchstaben sind
	\node (Stadt) at (2,2) {City};
	\node (Name) at (0,1) {Name};
	\node (Dortmund) at (0,0) {Dortmund};
	\node (Wetter) at (2,1) {Weather};
	\node (weather) at (2,0) {25\textcelsius, sun};
	\node (Veranstaltungen) at (4,1) {Events};
	\node (Indoor) at (3.3,0) {Indoor};
	\node (Kino) at (3.3,-1) {@operas};
	\node (Outdoor) at (4.7,0) {Outdoor};
	\node (Campuslauf) at (4.7,-1) {Marathon};
	
	\draw[-] (Stadt) -- (Name);
	\draw[-] (Stadt) -- (Wetter);
	\draw[-] (Stadt) -- (Veranstaltungen);
	\draw[-] (Name) -- (Dortmund);
	\draw[-] (Wetter) -- (weather);
	\draw[-] (Veranstaltungen) -- (Outdoor);
	\draw[-] (Veranstaltungen) -- (Indoor);
	\draw[-] (Outdoor) -- (Campuslauf);
	\draw[-] (Outdoor) -- (Campuslauf);
	\draw[-] (Indoor) -- (Kino);
	\draw[-] (Indoor) -- (Kino);
%	\boundingbox;
\begin{pgfonlayer}{background}
    \node [fill=blue!10,rounded corners,fit=(Stadt) (Dortmund) (Campuslauf)] {};
  \end{pgfonlayer}
	\end{tikzpicture}}
	}
\caption{An AXML document before and after the invocation of service calls\label{fig:AXML}}
\end{figure}

It might not be necessary to invoke all possible service calls. In the example of Figure~\ref{fig:AXML}, data about the weather might be relevant only if there are outdoor events and otherwise it does not need to be materialised. The choice which data needs to be materialised by the sender and the receiver may be influenced by considerations about performance, capabilities, security and functionalities and can be specified, for instance, by a DTD \cite{MiloAABN05}. An overview about AXML is given in \cite{AbiteboulBM08}.

The question whether a document can be rewritten so that it satisfies the specification then basically translates to the \emph{winning problem} for context-free games:  given a game and a string\footnote{The restriction to strings instead of trees was justified in \cite{MiloAABN05}.}, does Juliet have a winning strategy? 
 In general, this problem is undecidable, however it becomes decidable if Juliet has to follow a  \term{left-to-right}-strategy~\cite{MuschollSS06}. With such a strategy, Juliet basically traverses the string from left to right and decides, for each symbol, whether to play \emph{Read} (keep the symbol and go to the next symbol) or \emph{Call} (let Romeo replace the symbol). 
It turns out that, with this restriction, the set of strings, on which Juliet can win, is regular. Furthermore, she even has a uniform strategy that wins on all these strings and can be computed by a finite automaton (in a way that will be made more precise in Section~\ref{cha:existence}).  
 
With applications in mind, in which the AXML document comes as a data stream, Abiteboul, Milo and Benjelloun~\cite{AbiteboulMB05} initiated the study of a further strategy restriction, called \emph{one-pass strategies}:  Juliet still has to process the string from left-to-right, but now she does not even see the remaining part of the string, beyond the current symbol.

Due to the lack of knowledge of Juliet, one-pass strategies are more difficult to analyse and have less desirable properties than left-to-right strategies. For instance, in the sandbox game with one replacement rule $a\to b$ and the target language $\{ab,bc\}$, Juliet has a winning strategy that wins on the word $ab$ (Read the initial $a$) and one that wins on $ac$ (Call the initial $a$), but none that wins on both~\cite{AbiteboulMB05}. This example shows that even for some extremely simple games and input strings, there is no \emph{dominant} strategy\footnote{Such strategies were called \emph{optimum} in~\cite{AbiteboulMB05}. Undominated strategies were called \emph{optimal} there.} for Juliet, i.e., a strategy that wins on all words on which she has a winning strategy at all. However, both mentioned strategies are optimal in the sense that they can not be strictly improved; we call such strategies \emph{undominated}. In this paper, we consider a third kind of ``optimality'' which lies between the two former notions. We define a linear order $\lesl$ (the shortlex order) on strings and languages and call a strategy \emph{weakly dominant} if its winning set is maximal with respect to this order.  In the example above, assuming the usual order on the alphabet, the strategy which plays \Call on the initial $a$ is weakly dominant. 

Abiteboul et al.~\cite{AbiteboulMB05} also introduced \term{regular strategies}, a simple type of one-pass strategies defined by a finite state automaton, and therefore efficiently computable. In this paper, we also study a particularly simple form of regular strategies, called \emph{strongly regular}, computed by an automaton that is derived from the minimal automaton for the target language. We refer to Section~\ref{cha:preliminaries} for precise definitions of these notions.
\paragraph*{Contributions}

Since dominant strategies are so elusive, we broaden our view towards weakly dominant and undominated strategies and, at the same time, study conditions of games that guarantee such strategies. In particular, in Section~\ref{cha:existence}, we study the following three questions.
\begin{itemize}
\item Under which circumstances does a context-free game have a dominant, weakly dominant  or undominated one-pass strategy?
\item When can such a strategy even be chosen regular or strongly regular?
\item When can such a strategy be \emph{forgetful}, in that it does not need to remember all decisions it made, but only the (prefix of the) current string?
\end{itemize}

We identify various conditions on games that yield positive results, summarised in Theorem~\ref{theo:right}. The first one, the \emph{bounded-depth property}, is of a semantical nature and guarantees the existence of weakly dominant strategies. The other conditions are syntactical: \emph{prefix-free games}, \emph{non-recursive games}  and games with a finite target language have the bounded-depth property and therefore weakly dominant strategies. For the most natural\footnote{As explained later, every game can be transformed into a very similar prefix-free game.} condition, prefix-free games, there is always a regular weakly dominant strategy. If a non-recursive game or a game with a finite target language has a dominant strategy then it even has a strongly regular one. It remains open whether all context-free games have undominated strategies and, for that matter, whether there exist games which lack the bounded-depth property.

We complement these results by further negative results (Theorem~\ref{theo:wrong}): there exist games with a dominant strategy, but without a forgetful one; and there exist games with a regular and forgetful dominant strategy, but without a strongly regular one, and similarly for undominated strategies. Figure~\ref{fig:results} gives an illustration of our results of Section~\ref{cha:existence}.

\full{
In the second part, in Section~\ref{cha:complexity}, we determine the complexity of the following decision problems for regular strategies.
\begin{itemize}
\item Given a game, a regular strategy and a word, does the strategy win on the word?
\item Given a game and a word,  does a (strongly) regular winning strategy for the word exist?
\item Given two regular strategies, does one dominate the other?
\end{itemize}
The results are summarised in Theorem~\ref{thm:complexity}. 
} % \full
\short{
In Section~\ref{cha:complexity}, we determine the complexity of some algorithmic problems related to regular strategies.

Due to space constraints most proof details are delegated to the full version of this paper, which is attached as an appendix, for the convenience of the reviewers.
}

This paper is based on the Master's thesis of the first author, supervised by the other two authors~\cite{Coester15}. The thesis contains further results, some of which will be mentioned later. 

\paragraph*{Related work}

Further background about AXML is given in \cite{AbiteboulBM08,AbiteboulBCMM03,MiloAABN05}. Context-free games were introduced in~\cite{MuschollSS04}, which is the conference paper corresponding to~\cite{MuschollSS06}. The article studies the decidability and complexity of deciding whether a winning unrestricted or left-to-right strategy exists for a word in the general case and several restricted cases. One-pass strategies and (forgetful) regular strategies were introduced in~\cite{AbiteboulMB05}. The complexity of deciding, for a given context-free game, whether Juliet has a winning left-to-right strategy for every word for which she has a winning unrestricted strategy is studied in~\cite{BjoerklundSSK13}. Extended settings of context-free games with nested words (resembling the tree structure of (A)XML documents) are examined in~\cite{SchusterS15,Schuster16}.

\section{Preliminaries}\label{cha:preliminaries}

We denote the set of strings over an alphabet $\Sigma$ by $\Sigma^*$ and the set of non-empty strings by $\Sigma^+$. $\Sigma^k$ denotes the set of strings of length $k$ and $\Sigma^{\le k}$ the set of strings of length at most $k$.

A \emph{nondeterministic finite automaton (NFA)} is a tuple $\mathcal{A}=(Q,\Sigma,\delta,s,F)$, where $Q$ is the set of states, $\Sigma$ the alphabet, $\delta\subseteq Q\times\Sigma\times Q$ the transition relation, $s\in Q$ the initial state and $F\subseteq Q$ the set of accepting states. A \emph{run} on a string $w=w_1\cdots w_n$ is  a sequence $q_0,\ldots,q_n$ of states such that $q_0=s$ and, for each $i\le n$, $(q_{i-1},w_i,q_i)\in\delta$. A run is \emph{accepting} if $q_n\in F$. 
A word $w$ is in the language $L(\mathcal{A})$ of $\mathcal{A}$ if $\calA$ has an accepting run on $w$. 
If $\mathcal{A}$ is deterministic, \ie, for each $p$ and $a$, there is exactly one state $q$ such that $(p,a,q)\in\delta$, then we consider $\delta$ as \emph{transition function} $Q\times\Sigma\to Q$ and also use the \emph{extended transition function} $\delta^*:Q\times\Sigma^*\to Q$, as usual.
\paragraph*{Context-free games}
A \term{context-free game}, or a \term{game} for short, is a tuple $G=(\Sigma,R,T)$ consisting of
a finite alphabet $\Sigma$, a minimal\footnote{The assumption that $T$ is minimal will be convenient at times.} DFA $T=(Q,\Sigma,\delta,s,F)$, and 
a binary relation $R\subseteq\Sigma\times\Sigma^+$ such that for each $a\in\Sigma$, the \emph{replacement language} $L_a\df \{v\in\Sigma^+\mid (a,v)\in R\}$ of $a$ is regular.
We call $L(T)$ the \term{target language} of $G$. By $\Sigma_f=\{a\in\Sigma\mid \exists v\in\Sigma^+\colon (a,v)\in R\}$ we denote the set of \term{function symbols}, \ie the symbols occurring as the left hand side of a rule.
The languages $L_a$ are usually represented by regular expressions $R_a$, for each $a\in\Sigma_f$ and we specify $R$ often by expressions of the form $a\to R_a$.
We note that the definition of context-free games assures $\epsilon\not\in L_a$.

The semantics of context-free games formalises the intuition given in the introduction. In a configuration, we summarise the information about a current situation of a play together with some information about the history of the play. For the latter, let $\call{\Sigma_f}=\{\call{a}\mid a\in\Sigma_f\}$ be a disjoint copy of the set $\Sigma_f$ of function symbols, and let $\hist{\Sigma}=\Sigma\dotCup\call{\Sigma_f}$. A \term{configuration} is a tuple $(\alpha,u)\in\hist{\Sigma}^*\times\Sigma^*$. If $u$ is non-empty, \ie $u=av$ for $a\in\Sigma$ and $v\in\Sigma^*$, then we also denote this configuration by $(\alpha,a,v)$, consisting of a \term{history string} $\alpha$, a \term{current symbol} $a\in\Sigma$ and a \term{remaining string} $v\in\Sigma^*$. We denote the set of all (syntactically) possible configurations by $\calK$. Intuitively, if the $i$th symbol of the history string is $b\in\Sigma$ then this shall denote that Juliet's $i$th move was to read the symbol $b$, and if it is $\call{b}\in\call{\Sigma_f}$ then this shall denote that Juliet's $i$th move was to call $b$. The remaining string is the string of symbols that have not been revealed to Juliet yet.
By $\natural:\hist{\Sigma}^*\to\Sigma^*$ we denote\footnote{We usually omit brackets and write, e.g., $\natural \alpha\beta$ for $\natural(\alpha\beta)$.}  the homomorphism which deletes all symbols from $\call{\Sigma_f}$ and is the identity on $\Sigma$. We call $\delta^*(s,\natural\alpha)$ the \term{$T$-state} of the configuration $(\alpha,u)$.

A play is a sequence of configurations, connected by moves. In one move at a configuration  $(\alpha,a,v)$ Juliet can either ``read'' $a$ or ``call'' $a$. In the latter case, Romeo can replace $a$ by a string from $L_a$.
More formally, a \term{play} of a game is a finite or infinite sequence $\Pi=(K_0,K_1,K_2,\dots)$ of configurations with the following properties:
\begin{enumerate}[(a)]
\item The \term{initial configuration} is of the form $K_0=(\epsilon,w)$, where $w\in\Sigma^*$ is called the \term{input word}.
\item If $K_n=(\alpha,a,v)$, then either $K_{n+1}=(\alpha a,v)$ or $K_{n+1}=(\alpha\call{a},xv)$ with $x\in L_a$.\label{it:PreceedingConfiguration} In the former case we say that Juliet plays a \term{Read move}, otherwise she plays a \term{Call move} and Romeo replies by $x$.
\item If $K_n=(\alpha,\epsilon)$, then $K_n$ is the last configuration of the sequence. Its history string $\alpha$ is called the \term{final history string} of $\Pi$. Its \term{final string} is $\natural\alpha$. 
\end{enumerate}
 
A play is \term{winning for Juliet} (and \emph{losing for Romeo}) if it is finite and its final string is in the target language $L(T)$. A play is \term{losing for Juliet} (and \emph{winning for Romeo}) if it is finite and its final string is not in $L(T)$. An infinite play is neither winning nor losing for any player.

\paragraph*{Strategies}\label{sec:JulietStrategies}

As mentioned in the introduction, we are interested in so-called one-pass strategies for Juliet, where Juliet's decisions do not depend on any symbols of the remaining string beyond the current symbol.

A \term{one-pass strategy for Juliet} is a map $\sigma\colon\hist{\Sigma}^*\times\Sigma_f\to\{\Call,\Read\}$, where the argument corresponds to the first two components of a configuration. A \term{strategy for Romeo} is a map $\tau\colon\sigma\colon\hist{\Sigma}^*\times\Sigma_f\to\Sigma^+$ where $\tau(\alpha,a)\in L_a$ for each $(\alpha,a)\in\hist{\Sigma}^*\times\Sigma_f$.\footnote{Even though we think of Romeo as an omniscient adversary, it is not necessary to provide the remaining string as an argument to $\tau$: The remaining string is uniquely determined by the input word and his own and Juliet's previous moves.} We generally denote strategies for Juliet by $\sigma, \sigma', \sigma_1, \ldots$ and Romeo strategies by $\tau, \tau', \tau_1, \ldots$.
We often just use the term \emph{strategy} to refer to a one-pass strategy for Juliet.

The (unique) \term{play of $\sigma$ and $\tau$ on $w$} is a play $\Pi(\sigma,\tau,w)=(K_0,K_1,K_2,\dots)$ with input word $w$ satisfying that
\begin{itemize}
\item if $K_n=(\alpha,a,v)$ and $\sigma(\alpha,a)=\Read$, then $K_{n+1}=(\alpha a,v)$,
\item if $K_n=(\alpha,a,v)$ and $\sigma(\alpha,a)=\Call$, then $K_{n+1}=(\alpha\call{a},\tau(\alpha,a)v)$.
\end{itemize} The \emph{depth} of a finite play is its maximum nesting depth of \Call moves. E.g., if Romeo replaces some symbol $a$ by a string $u$ and Juliet calls a symbol in $u$, the nesting depth of this latter \Call move is 2.

A strategy $\sigma$ is \term{terminating} if each of its plays is finite. The \term{depth} of $\sigma$ is the supremum of depths of plays of $\sigma$. Note that each strategy with finite depth is terminating. The converse, however, is not true and it is easy to construct counter-examples of a game and a strategy $\sigma$ where each play of $\sigma$ has finite depth but depths are arbitrarily large.

A strategy $\sigma$ \term{wins} on a string $w\in\Sigma^*$ if every play of $\sigma$ on $w$ is winning (for Juliet). By $W(\sigma)=W_G(\sigma)$ we denote the set of words on which $\sigma$ wins in $G$. In contrast, $\sigma$ \term{loses} on $w$ if there exists a losing play of $\sigma$ on $w$. Note that $\sigma$ neither wins nor loses on $w$ if there exists an infinite play of $\sigma$ on $w$ but no losing play of $\sigma$ on $w$.

A strategy $\sigma$ \term{dominates} a strategy $\sigma'$  if $W(\sigma')\subseteq W(\sigma)$. It \term{strictly dominates} $\sigma'$, if $W(\sigma')\subsetneq W(\sigma)$.
A strategy $\sigma$ is \term{dominant} if it dominates every other (one-pass) strategy. It is \term{undominated} if there is no other strategy that strictly dominates it.

We also consider a third, intermediate, type of optimality. To define it, we fix some total order $<$ of the alphabet
        $\Sigma$. We order strings by \emph{shortlex order}, \ie for
        two strings $v,w\in\Sigma^*$ we define $v\lsl w$ if $|v|<|w|$
        or if $|v|=|w|$ and $v$ precedes $w$ in the lexicographical
        order. We extend this to a total order $\lesl$ on sets of
        words as follows. Let $V,W\subseteq\Sigma^*$ be two sets with
        $V\neq W$. Their order is determined by the minimal string $w$
        (with respect to shortlex order $\lesl$) that is contained in
        only one of the two sets. If $w\in W$, then $V\lsl W$;
        otherwise $W\lsl V$. We observe that if $V\subsetneq W$ then
        $V\lsl W$.
A strategy $\sigma$ is \emph{weakly dominant} if, for every 
strategy $\sigma'$ it holds $W(\sigma')\lesl W(\sigma)$. Thus, a weakly dominant strategy can be seen as a best undominated strategy with respect to \lesl.  

A strategy $\sigma$ is \emph{regular} if the set $L$ of strings $\alpha a$ with $\sigma(\alpha,a)=\Call$ is regular. In this case, a DFA $\calA$ for $L$ is called a \emph{strategy automaton} for $\sigma=\sigma_\calA$.
 A strategy is \term{forgetful} if its decisions are independent of symbols from $\call{\Sigma_f}$ in the history string, \ie if $\sigma(\alpha,a)=\sigma(\beta,a)$ whenever $\natural\alpha=\natural\beta$.
It is not hard to show that a strategy is regular \emph{and} forgetful if and only if $L'\df \{\natural\alpha a\mid \sigma(\alpha,a)=\Call\}$ is regular. A DFA $\calA$ for $L'$ is also called a strategy automaton, and we write $\sigma_{\mathcal A}=\sigma$ again.

We are particularly interested in the special case of regular forgetful strategies where Juliet's decisions depend only on the current $T$-state and the current symbol. More precisely, if $T=(Q,\Sigma,\delta,s,F)$ is the target automaton and the strategy automaton is of the form $\calA=(Q\cup\{\Call\},\Sigma,\delta_\calA,s,\{\Call\})$ with $\delta_\calA(q,a)\in\{\delta(q,a),\Call\}$, for each $q$ and $a$, then $\sigma_\calA$ is called \emph{strongly regular}.

The following lemma yields a convenient property of one-pass strategies.
\begin{lemma}\label{lem:alwaysTerm}
	For each game $G$ and strategy $\sigma$ there is a terminating strategy $\sigma'$ s.t.\ $W(\sigma)\subseteq W(\sigma')$.
\end{lemma}
\full{
	\begin{proof}
		The strategy $\sigma'$ plays like $\sigma$ except that once it is apparent that Romeo can force an infinite play against $\sigma$, all remaining moves of $\sigma'$ are Read.
		
		For a formal definition, consider some $\alpha\in\hs$ and $a\in\Sigma_f$. If there exists a play $\Pi$ of $\sigma$ that contains a configuration $(\alpha,a,v)$ for some $v\in\Sigma^*$ such that no later configuration with remaining string $v$ occurs in $\Pi$, then let $\sigma'(\alpha,a)=\Read$ and $\sigma'(\alpha a\beta,b)=\Read$ for each $\beta\in\hs$ and $b\in\Sigma_f$. For all elements of the domain for which $\sigma'$ is not already defined by this, we define $\sigma'$ like $\sigma$. Clearly, $\sigma'$ is terminating and $W(\sigma)\subseteq W(\sigma')$.
	\end{proof}
}
\section{Existence of dominant and undominated strategies}\label{cha:existence}
One-pass strategies are a restriction of \emph{left-to-right} strategies, where Juliet's moves can depend on the whole current configuration, including the remaining string.  
Every context-free game has a dominant left-to-right strategy for Juliet, \ie a strategy that dominates all other left-to-right strategies \cite{MuschollSS04,AbiteboulMB05}. Furthermore, there is such a strategy whose moves can be computed by a DFA that reads the current $T$-state, the current symbol and the remaining string.
Due to the ignorance of Juliet regarding the remaining string, this result fails to hold for one-pass strategies. It was shown in \cite[Example 2.6]{AbiteboulMB05} that dominant (one-pass) strategies need not exist, even for very restricted games.

\tikzset{action/.style={
      %  The shape:
      rectangle,
      %  The size:
      minimum size=6mm,
      minimum height = .8cm,
      %  The border:
     % very thick,
      % draw=red!50!black!50,
      % Corners
      rounded corners=2mm,
      %  The filling:
      fill=yellow,
    }
}

\newcommand{\mybox}[4]{\node[action,fill=yellow] (#1) at (#2) {\begin{minipage}{#3}%
\centering #4\end{minipage}};}

\begin{figure}[ht]
\scalebox{0.8}{
\begin{tikzpicture}
\node (solid1) at (0,0) {};
\node (solid2) at (1,0) {};
\draw[->] 	(solid1) -- (solid2) node [right] {general;};
\end{tikzpicture}
}
\scalebox{0.8}{
\begin{tikzpicture}
\node (sdash1) at (0,0) {};
\node (sdash2) at (1,0) {};
\draw[->,densely dashed] 	(sdash1) -- (sdash2) node [right] {bounded depth property;};
\end{tikzpicture}
}
\scalebox{0.8}{
\begin{tikzpicture}
\node (ldash1) at (0,0) {};
\node (ldash2) at (1,0) {};
\draw[->,loosely dashed] 	(ldash1) -- (ldash2) node [right] {prefix-free;};
\end{tikzpicture}
}
\scalebox{0.8}{
\begin{tikzpicture}
\node (dotted1) at (0,0) {};
\node (dotted2) at (1,0) {};
\draw[->,dotted] 	(dotted1) -- (dotted2) node [right] {finite $L(T)$ or non-recursive};
\end{tikzpicture}
}

\subcaptionbox{Dominant strategies}[.46\textwidth]{
\scalebox{0.7}{
\begin{tikzpicture}
\mybox{un}{1,2.5}{1.8cm}{unrestricted}
\mybox{reg}{4,2.5}{1.8cm}{regular}
\mybox{forg}{1,1}{1.8cm}{forgetful}
\mybox{forgreg}{4,1}{1.8cm}{forgetful,\\[-1mm]regular}
\mybox{sreg}{7,1}{1.8cm}{strongly\\[-1mm]regular}

\node (dummy) at (1,3.7) {};

\draw[->]  (dummy) -- (un)  node [midway] {$-$};
\draw[->]  (un) -- (forg) node [midway] {$-$};
\draw[->]  (reg) -- (forgreg) node [midway] {$-$};
\draw[->]  (reg) -- (forg) node [midway] {\footnotesize$\backslash$};
\draw[->]  (forgreg) -- (sreg) node [midway] {\footnotesize$|$};
\draw[->,dotted]  (un) -- (sreg);
\end{tikzpicture}
}
}
\hfill
\subcaptionbox{Undominated strategies}[.46\textwidth]{
\scalebox{0.7}{
\begin{tikzpicture}
\mybox{un}{0,2.5}{1.8cm}{unrestricted}
\mybox{reg}{3,2.5}{1.8cm}{regular}
\mybox{forg}{0,1}{1.8cm}{forgetful}
\mybox{forgreg}{3,1}{1.8cm}{forgetful,\\[-1mm]regular}
\mybox{sreg}{6,1}{1.8cm}{strongly\\[-1mm]regular}

\node (dummy) at (0,3.7) {};
\node (dummy1) at (3,3.7) {};

\draw[->,densely dashed]  (dummy) -- (un);
\draw[->,loosely dashed]  (dummy1) -- (reg);
\draw[->]  (reg) -- (forgreg) node [midway] {$-$};
\draw[->]  (un) -- (forg) node [midway] {$-$};
\draw[->]  (reg) -- (forg) node [midway] {\footnotesize$\backslash$};
\draw[->]  (forgreg) -- (sreg) node [midway] {\footnotesize$|$};
\end{tikzpicture}
}
}

\caption{Illustration of existence and implication results with respect to dominant and undominated strategies. Arrows from above indicate unconditional existence, arrows between boxes indicate implication of existence. Crossed out arrows indicate negative results. Arrow types indicate results for restricted game types as shown above. Upward and right-to-left solid arrows are not drawn, they all exist by definition. }
\label{fig:results}
\end{figure}
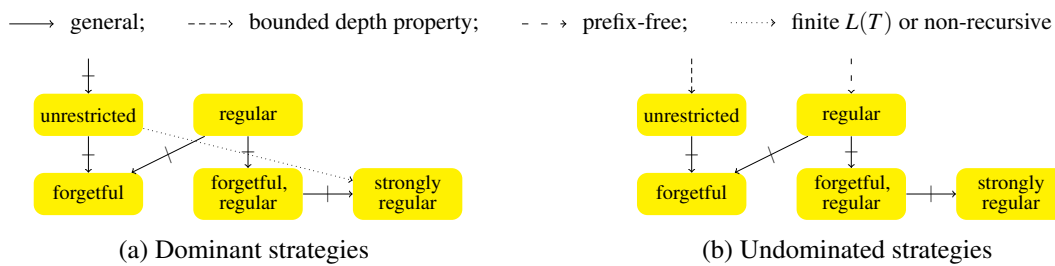

In this section, we investigate dominant and undominated strategies more deeply. The main results are illustrated in Figure~\ref{fig:results}.  We show that the situation for undominated strategies is much better than for dominant ones. Although it remains unclear whether such strategies exist for \emph{all} games, we identify important classes of games in which they do exist. More precisely, we give a semantical restriction, the bounded depth property, that guarantees existence of undominated strategies and a fairly natural syntactical restriction, prefix-freeness, that guarantees existence of regular undominated strategies.  For two other families of restricted games, non-recursive games and games with finite target language, we show that if they have a dominant strategy, there is even a strongly regular one. Finally, games with a unary alphabet also guarantee undominated strategies.

We complement these positive results by some negative findings: in general, the existence of a (regular) dominant strategy does not imply that there is a forgetful one. Furthermore, the existence of a regular and forgetful dominant strategy does not imply that there is a strongly regular one. The same holds with respect to undominated strategies, even for very restricted games.

\begin{definition}\label{def:boundedDepthProperty}
  \begin{itemize}
  \item A game $G=(\Sigma,R,T)$ has the \term{bounded depth property} if there
    exists a sequence $(B_k)_{k\in\mathbb N_0}\subseteq \mathbb N$
    such that for each one-pass strategy $\sigma$ for $G$ and each
    $k\in\mathbb N$ there exists a one-pass strategy $\sigma_k$ that
    wins on each $w\in W(\sigma)\cap\Sigma^{\le k}$ with plays
    of depth at most $B_{\abs{w}}$.
  \item A game $G=(\Sigma,R,T)$ is \emph{prefix-free} if each replacement language $L_a$ is prefix-free, that is, there are no $u,v\in L_a$ where $u$ is a proper prefix of $v$. 
  \item A game $G=(\Sigma,R,T)$ is \term{non-recursive} if no symbol can be derived from itself by a sequence of rules, \ie there do not exist $a_0,\dots,a_n\in\Sigma_f$, $n\ge1$, such that $a_0=a_n$ and for each $k=1,\dots,n$ there exists a word in $L_{a_{k-1}}$ containing $a_k$. 
  \item A game $G=(\Sigma,R,T)$ is \emph{unary} if $|\Sigma|=1$.
  \end{itemize}
\end{definition}

\subsection{Positive results}

Our positive results are summarised in the following theorem.

\begin{theorem}\label{theo:right}
  \begin{enumerate}[(a)]%[{label={(\alph*)}}]
  \item Every game with the bounded depth property has a  weakly dominant strategy.
  \item Every prefix-free game has a regular weakly dominant strategy.
  \item Every game with a finite target language that has a dominant strategy has a strongly regular dominant strategy. 
  \item Every non-recursive game with a dominant strategy has a strongly regular dominant strategy. 
  \item Every unary game has an undominated strategy.
  \end{enumerate}
\end{theorem}
Since every non-recursive game trivially has the bounded-depth property, it also has an undominated strategy by part (a). Games with a finite target language constitute a  further class with the bounded-depth property and with regular undominated strategies \cite[Corollary 4.15]{Coester15}.

A concept used in the proofs of parts (a) and (d) is the convergence of a sequence of one-pass strategies.

\begin{definition}
	A sequence $(\sigma_k)_{k\in\mathbb N}$ of strategies \term{converges to} a strategy $\sigma$ if for each $n\in\mathbb N$ there exists $k_0\in\mathbb N$ such that for each $k\ge k_0$ and $(\alpha,a)\in\hist{\Sigma}^{\le n}\times\Sigma$ it holds that $\sigma(\alpha,a)=\sigma_k(\alpha,a)$.
\end{definition}

\full{
	This is the same as the convergence in the metric space of strategies where the metric $d$ is given by $d(\sigma,\sigma)=0$ and $d(\sigma,\sigma')=\frac{1}{n}$ for $\sigma\ne\sigma'$, where $n$ is minimal such that there exists $(\alpha,a)\in\hist{\Sigma}^n\times\Sigma$ and $\sigma(\alpha,a)\ne\sigma'(\alpha,a)$.
} % \full

\begin{lemma}[Convergence Lemma]\label{lemma:conv}
	Let $G$ be a game and $(\sigma_k)_{k\in\mathbb N}$ be a sequence of one-pass strategies that converges to some one-pass strategy $\sigma$. Let  $L_1\subseteq L_2 \subseteq \cdots$ be an infinite sequence of languages such that, for every $k$, $L_k\subseteq W(\sigma_k)$ and let $L\df \displaystyle \bigcup_{k\in\mathbb N} L_k$. 
	
	Then $\sigma$ wins on every word $w\in L$ for which it terminates.
\end{lemma}
\full{
	\begin{proof}
		Towards a contradiction, suppose that $\sigma$ loses on a word $w\in L$. Then there exists a strategy of Romeo with which he wins the (finite) play $\Pi=\Pi(\sigma,\tau,w)=(K_0,\dots,K_n)$. Let $k_0\in\mathbb N$ be such that for each $k\ge k_0$ and $(\alpha,a)\in\hist{\Sigma}^{\le n}\times\Sigma$ it holds that $\sigma(\alpha,a)=\sigma_{k}(\alpha,a)$. Let furthermore $k_1$ be such that $w\in L_{k_1}$ and let $k\df \max(k_0,k_1)$.
		Then $\Pi$ is also a play of $\sigma_{k}$ on $w$. But then $\sigma_k$ loses on $w\in L_{k_1}\subseteq L_k$, the desired contradiction.
	\end{proof}
} % \full

\short{In the next section, we give the proof of part (a). In the subsequent section, we sketch the proofs of parts (b), (c) and (d). We also explain our motivation for prefix-freeness, which we consider the most relevant case and with the most interesting proof.}

\subsubsection{Games with the bounded depth property}\label{sec:boundedDepth}

Before we prove Theorem~\ref{theo:right}~(a), we define some further notation.
For a strategy $\sigma$ and some $i\ge0$, we denote by $\restr{\sigma}{i}$ the restriction of $\sigma$ to the first $i$ rounds of the game. Thus $\restr{\sigma}{i}$ is a mapping $\restr{\sigma}{i}\colon \hist{\Sigma}^{\le i-1}\times\Sigma\to\{\Call,\Read\}$ and $\restr{\sigma}{0}$ is the mapping with empty domain.

\begin{proofof}{Theorem~\ref{theo:right}~(a)}
    Let $G$ be a context-free game with the {bounded depth property} and let $(B_k)_{k\in\mathbb N_0}\subseteq \mathbb N$ be its sequence of depth bounds.

We first define a language $L$ which will serve as the winning set of
the weakly dominant strategy that will be constructed below.

The definition of $L$ is by induction. For each $k\ge 0$, we define a set $L_k\subseteq \Sigma^{\le k}$ such that $L_k\subseteq L_{k+1}$, and finally let $L\df \bigcup_k L_k$.

Let $L_0=\{\epsilon\}$ if $\epsilon$ is in the target language of $G$, and $L_0=\emptyset$ otherwise.

For $k\ge 0$, we define $L_{k+1}$ as the maximal set with respect to $\lesl$ of the form $W(\sigma)\cap\Sigma^{\le k+1}$ for some strategy $\sigma$ with $L_k\subseteq W(\sigma)$.
It is easy to see that the following two properties hold by construction.
\begin{enumerate}[(1)]%[{label={(\arabic*)}}]
\item For each $k$, there is a strategy $\sigma$ such that $L_k\subseteq W(\sigma)$.
\item There is no strategy $\sigma$ with $L\lsl W(\sigma)$.
\end{enumerate}
Thanks to property (2), it suffices to construct a strategy $\shat$ with $L\subseteq W(\shat)$

For each $k\ge 0$, let $S_k$ be the set of strategies $\sigma$ with $L_k\subseteq W(\sigma)$ for which each play on a word $w\in L_k$ has depth at most $B_{|w|}$.  Because of property (1) and since $G$ has the bounded depth property, we have $S_k\not=\emptyset$ for every $k\ge 0$.
    
    We will construct mappings $\rho_k\colon \hist{\Sigma}^{\le k-1}\times\Sigma\to\{\Call,\Read\}$ such that for every $k\ge 0$,
    \begin{itemize}
    	\item $\rho_{k+1}$ extends $\rho_k$; more
    	precisely: $\restr{\rho_{k+1}}{k}=\rho_k$, and
    	\item for each $\ell\ge k$ there exists $\sigma^k_\ell\in S_\ell$ with $\rho_k=\restr{\sigma^k_\ell}{k}$.
    \end{itemize}
	Let $\rho_0$ be the mapping with empty domain. Fix $k$ such that $\rho_0,\dots,\rho_k$ are defined and have the stated properties. Since there are only finitely many mappings $\hist{\Sigma}^{\le k}\times\Sigma\to\{\Call,\Read\}$, one of them has to occur infinitely often within $\big(\restr{\sigma^k_\ell}{k+1}\big)_{\ell\ge k}$. Let $\rho_{k+1}$ be such a mapping. For $\ell'\ge k+1$ we can choose $\sigma^{k+1}_{\ell'}=\sigma^{k}_{\ell}$ for some $\ell\ge\ell'$ with $\rho_{k+1}=\restr{\sigma^k_\ell}{k+1}$. This defines a sequence $\rho_0,\rho_1,\dots$ with the properties above. Let $\shat$ be the strategy that is uniquely determined by $\restr{\shat}{k}=\rho_k$, for every $k$. Clearly $(\rho_k)_{k\in\mathbb N}$ converges\footnote{Since the $\rho_k$ are only partially defined, one might consider the strategies $\sigma_k$ that result from the $\rho_k$ which take the value \Call whenever $\rho_k$ is undefined.} to $\shat$. 

Thanks to Lemma~\ref{lemma:conv} it suffices to show that $\shat$ terminates on $L$.
Let thus $w\in L$ and $\tau$ be a Romeo strategy. We show that the depth of  $\Pi\df\Pi(\shat,\tau,w)$ is at most $B_{|w|}$. Otherwise let $k$ be such that in the $k$th round Juliet does a $\Call$ move of nesting depth $B_{|w|}+1$. However, $\restr{\shat}{k}=\rho_k=\restr{\sigma_\ell^k}{k}$, where $\ell=\max\{k,|w|\}$, and $\sigma_\ell^k\in S_\ell$ has depth at most $B_{|w|}$ on $w\in L_\ell$, a contradiction. Therefore the depth of $\Pi$ is at most $B_{|w|}$ and by K\"{o}nig's Lemma $\Pi$ is thus finite, completing the proof.
\end{proofof}

\subsubsection{Prefix-free games}\label{sec:selfDelim}

Prefix-freeness appears as a realistic constraint for a practical
(Active XML) setting since it can be easily enforced by suffixing each
replacement string with a special end-of-file symbol. In this sense,
\emph{every} game $G=(\Sigma,R,T)$ can be transformed into a
prefix-free game $G'=(\Sigma',R',T')$ by letting
$\Sigma'=\Sigma\dotCup\{\$\}$ for some new symbol $\$\notin\Sigma$
that shall denote the end of replacement strings, and further letting
$R_a'=R^{\phantom\prime}_a\$$ for each $a\in\Sigma_f$ to enforce that
replacement words end with $\$$, and adding a loop transition for the
symbol $\$$ to each state of $T$ (accomplishing that the symbol~$\$$
is ``ignored'' by the target language). Another special case of
prefix-free games, which is similar to the \emph{one-pass with size}
setting discussed in \cite{AbiteboulMB05}, are games where the
alphabet $\Sigma$ contains (besides other symbols) numbers $1,\dots,N$
for some $N\in\mathbb N$ and all replacement strings are of the form
$nx$ where $x\in\Sigma^+$ and $n=\abs{x}$. Our result for prefix-free games also easily transfers to the setting where the input word is revealed to Juliet in a one-pass fashion, but Romeo's replacement words are revealed immediately.

We \short{sketch}\full{give} the proof of Theorem~\ref{theo:right}~(b) in the following.

A context-free game on a string $w=a_1\cdots a_n$ can be viewed as a
sequence of $n$ games on the single symbols $a_1,\ldots,a_n$.
Intuitively, in prefix-free games Juliet has the benefit to know when
a subgame on some symbol $a_i$ has ended and when the next subgame
starts. 

This allows us to view strategies of Juliet in a hierarchical
way: they consist of a top-level strategy that chooses, whenever a
subgame on some $a_i$ starts, a strategy for this subgame. This
choice may take the current history string into account.
We will use this view to proceed in an inductive fashion: we 
establish that there are automata for the subgame strategies and then combine these
automata with suitable automata for a ``top-level'' strategy. 

It turns out that the choice of the top-level strategy boils down to
an ``online word problem'' for NFAs which we introduce and study first.

\paragraph*{The online word problem for NFAs}

In the online-version of the word problem for an NFA $\calN$, denoted \online, the
single player gets to know the symbols of a word one by one, and always needs to decide which transition $\calN$ should
take before the next symbol is revealed. We only consider the case that $\calN=(Q,\Sigma,\delta,s,F)$ has at
least one transition for each symbol from each state. Formally, a strategy is a map $\rho\colon\Sigma^*\to Q$ such that $\rho(\epsilon)=s$ and $(\rho(w),a,\rho(wa))\in\delta$ for each $w\in\Sigma^*$ and $a\in\Sigma$.

Given a strategy $\rho$
for \online, we denote by $W_\calN(\rho)$ the \emph{winning
  set} of words that are accepted by $\calN$ if the player follows
$\rho$. A strategy $\rho$ is
\emph{weakly dominant}  if, for every
strategy $\rho'$, it holds $W_\calN(\rho')\lesl W_\calN(\rho)$.

\full{
For an NFA $\calN=(Q,\Sigma,\delta,s,F)$ and a state $q\in Q$, we
write $\calN^q$ for the NFA $(Q,\Sigma,\delta,q,F)$ with starting
state $q$.
We say that a word $w$ is \emph{universally accepted} by an NFA
$\calN$, if \emph{every} run of $\calN$ on $w$ is accepting. 
}%\full

We are interested in strategies that can be computed by automata. A
particularly simple such strategy for \textsc{OnlineNFA}($\mathcal
N$) can be obtained by transforming $\mathcal N$ into a DFA $\mathcal
D$ by removing transitions. The associated strategy $\rho_{\mathcal
  D}$ is the one that only uses the transitions of $\calD$. We prove that $\calD$ can be chosen such that $\rho_{\mathcal
  D}$ is weakly dominant. 

\begin{lemma}\label{lem:OnlineNFA}
	For each NFA $\mathcal N$, there exists a DFA $\mathcal D$
        obtained by removing transitions from $\mathcal N$ such that
        $\rho_{\mathcal D}$ is a weakly dominant strategy for \online.
\end{lemma}
\full{
\begin{proof}	
	We define a sequence $\mathcal N_0,\mathcal N_1,\dots$
        of NFAs, starting from $\mathcal N_0=\mathcal N$. Each
        $\calN_{n+1}$ results from $\calN_n$ by removing 
        transitions in the following way. For a state $q$, let
        $L^q_n$ be the set of words of length at most $n$ that are
        non-deterministically accepted by $\mathcal N^q_n$. For each state $q$ and symbol $a$, we only keep
        the transitions to those successor states $p$ that maximise
        $L^{p}_n$ with respect to $\lesl$. 
	
	Since there are only finitely many transitions in $\calN$, the
        sequence $\mathcal N_0,\mathcal N_1,\dots$ becomes stationary
        after finitely many steps. By $\widetilde{\calN}$ we denote the last NFA
        of the sequence.

	\begin{myclaim}\label{cl:LSnAlsoUniv}
	For each $q$ and $n$, all words in $L^q_n$ are
        non-deterministically and universally accepted by $\widetilde{\calN}^q$.
	\end{myclaim}
	\begin{proof}
		We show a slightly stronger claim by induction on $n$:
                for each $q$, $n$ and $m\ge n$, all words in $L^q_n$ are
        non-deterministically and universally accepted by $\calN^q_m$.
                 For $n=0$ the statement is trivial.
		
		Let $aw\in L^q_{n+1}$ with $a\in\Sigma$ and
                $w\in\Sigma^{\le n}$. Let $p$ be a state reached
                after reading the initial $a$ in some accepting run of
                $\calN^q_{n+1}$ on
                $aw$. Then
                $w\in L^p_n$. By definition of $\mathcal N_{n+1}$,
                $L^r_n=L^p_n$ for every state $r$ reachable
                from $q$ by reading $a$ in $\mathcal
                N_{n+1}$. Therefore, by the induction
                hypothesis, $w$
                is universally accepted by $\calN^r_n$, for any such $r$. Thus,
                $aw$ is universally accepted by $\calN^q_{n+1}$.

	Since each $\calN_{i+1}$ results from $\calN_{i}$ by removing
        transitions, every word that is universally accepted by
        $\calN_{i}$ is also universally accepted by
        $\calN_{i+1}$. On the other hand, since each state has a
        transition for each symbol, every universally accepted word is
        also nondeterministically accepted. 
Therefore, for each $q$, $n$, all words in $L^q_n$ are
non-deterministically and universally accepted by
$\widetilde{\calN}^q_m$, for every $m\ge n$, and therefore also by $\calN^q$. 
	\end{proof}

	\begin{myclaim}\label{cl:LSnMax}
For each strategy $\rho$ for \online, $W_\calN(\rho)\lesl  L(\widetilde{\calN})$.
	\end{myclaim}
	\begin{proof}
We show a more precise statement: for each $q$,
each $n\ge 0$,  and
each strategy $\rho$ for \online[\calN^q],
$W_{\calN^q}(\rho)\cap\Sigma^{\le n}\lesl  L^q_n$.
Claim~\ref{cl:LSnMax} easily follows from this statement and Claim~\ref{cl:LSnAlsoUniv}.

Towards a contradiction, assume that there are $n,q,\rho$
with $W_{\calN^q}(\rho)\cap\Sigma^{\le n}\gsl L^q_n$. Let $n$ be minimal
with this property and let $w\in (W_{\calN^q}(\rho)\cap\Sigma^{\le
  n})\setminus L^q_n$ be chosen minimal with respect to
$\lesl$. Clearly, for all $w'\lsl w$ it holds $w'\in (W_{\calN^q}(\rho)\cap\Sigma^{\le
  n}) \Leftrightarrow w'\in L^q_n$

Clearly $n>0$, so we
can write $w=au$ for $a\in\Sigma$ and $u\in\Sigma^{n-1}$. Let $p=\rho(a)$. Since $n$ is minimal, for the
strategy $\rho'$ induced by $\rho$ for $\calN^p$, defined by $\rho'(v)=\rho(av)$, it holds
\begin{align}\label{eq:prefixfreeTyp2UndomExists1}
W_{\calN^p}(\rho')\cap\Sigma^{\le n-1}\lesl  L^p_{n-1}. 
\end{align}

Let $r$ be some state for which $(q,a,r)$ is a transition in
$\calN_n$. We claim
\begin{align}\label{eq:prefixfreeTyp2UndomExists2}
W_{\calN^p}(\rho')\cap\Sigma^{\le n-1}\gsl  L^r_{n-1}. 
\end{align}
Towards a contradiction assume otherwise. Since $u\in W_{\calN^p}(\rho')\cap\Sigma^{\le n-1}$ and $u\not\in
L^r_{n-1}$, there must be some  $u'\lsl u$ with $u'\not\in W_{\calN^p}(\rho')\cap\Sigma^{\le n-1}$ and $u'\in
L^r_{n-1}$. However, then $au'\not\in W_{\calN^q}(\rho)\cap\Sigma^{\le
  n}$ and $au'\in L^q_n$, a contradiction, since $au'\lsl au=w$.

(\ref{eq:prefixfreeTyp2UndomExists1}) and
(\ref{eq:prefixfreeTyp2UndomExists2}) imply that $(q,a,p)$ is not a transition in
$\calN_n$. 
	Let $m\le n-1$ be maximal such that $(q,a,p)$ is a transition in
        $\mathcal N_m$. Then $L^p_m\lsl L^r_m$ and hence
        $L^p_{n-1}\lsl L^r_{n-1}$, contradicting (\ref{eq:prefixfreeTyp2UndomExists1}) and
(\ref{eq:prefixfreeTyp2UndomExists2}). 
	\end{proof}

	 Let
        $\calD$ be any DFA resulting from $\widetilde{\calN}$ by removing some
        further transitions. By Claim~\ref{cl:LSnAlsoUniv},
        $L(\calD)=L(\widetilde{\calN})$. Therefore, by Claim~\ref{cl:LSnMax},
        $\rho_\calD$ is weakly dominant, concluding the proof of Lemma~\ref{lem:OnlineNFA}.
\end{proof}
} % \full

\paragraph*{Game composition and game effects.}

Let in the following, $G=(\Sigma,R,T)$ be a prefix-free game with
        $T=(Q,\Sigma,\delta,s,F)$. Let furthermore, for every
        $a\in\Sigma$, 
        $\calA_a=(Q_a,\Sigma,\delta_a,s_a,\{f_a\})$ be a minimal DFA for the replacement language $L_a$ of $a$. Since $L_a$ is
        prefix-free, $\calA_a$ has a unique accepting state $f_a$.

For a strategy $\sigma$ (for Juliet or Romeo) and a string $\alpha\in\hist{\Sigma}^*$, we define the \emph{substrategy} $\sigma^\alpha$ of $\sigma$ by $\sigma^\alpha(\beta,a)=\sigma(\alpha\beta, a)$.
 
In the following,
        $\states{q}{w,\sigma}$ denotes the set of $T$-states that can be
        reached at the end of a play of $\sigma$ on $w$ if the initial
        state of $T$ were $q$. More precisely, it is the
set of states of the form
        $\delta^*(q,\natural \alpha)$ where $\alpha$ is a final history string of a play of $\sigma$ on $w$. 

An \emph{effect triple} $(p,a,S)$ consists of a state $p\in Q$, a
symbol $a\in \Sigma$ and a set $S\subseteq Q$. We say that $(p,a,S)$
is an effect triple of $\sigma$ if $\states{p}{a,\sigma}\subseteq S$. We call $(p,a,S)$ \emph{trivial} if $\delta(p,a)\in S$, \ie if it is an effect triple of a strategy that plays Read on $a$.
The \emph{single-symbol effect} $\sse(\sigma)$ of a strategy $\sigma$
is the set of all its effect triples. Finally, we define the
\emph{effect set} $E(\sigma)$ of a
strategy $\sigma$ as $\displaystyle E(\sigma)\df \bigcup_{\alpha\in\hist{\Sigma}^*} \sse(\sigma^\alpha)$. That is, $E(\sigma)$
contains all effect triples that are induced by
substrategies of $\sigma$.

With a set $E$ of effect triples we associate an NFA
$\calN_E=(\calP(Q),\Sigma,\delta_E,\{s\},\calP(F))$, where
$\delta_E$ is defined as follows. For sets $S,S'\subseteq Q$ and
$a\in\Sigma$, $(S,a,S')\in\delta_E$ if, for each $p\in S$, there is
some $S''\subseteq S'$ such that $(p,a,S'')\in E$. 

\begin{proposition}\label{prop:GameToOnline}
Let $G=(\Sigma,R,T)$ be a prefix-free game, $E$ a set of effect triples, and $\sigma$ a
terminating strategy such that $E(\sigma)\subseteq E$. Then there is a strategy
$\rho$ for $\online[\calN_E]$ such that $W_G(\sigma)=W_{\calN_E}(\rho)$.
\end{proposition}
%\full{
\begin{proof}
  It is straightforward to verify that $\rho(w)\df \states{s}{w,\sigma}$ yields a well-defined strategy $\rho$ for $\online[\calN_E]$. The proposition follows, since $w\in W_{\calN_E}(\rho)$ if and only if $\rho(w)\subseteq F$, and $w\in W_G(\sigma)$ if and only if $\states{s}{w,\sigma}\subseteq F$.
\end{proof}
%} % \full

We say that a strategy automaton $\calA=(Q_\calA,\hist{\Sigma},\delta_\calA,s_\calA,F_\calA)$ is
\emph{$(p,a,S)$-inducing} if $\sigma_\calA$ is terminating, $a\in \Sigma_f$, and the following conditions hold.
\begin{itemize}
\item For each $u\in L_a$, $\states{p}{u,\sigma_\calA}\subseteq S$.
\item There are disjoint subsets $Q_{\calA,q}\subseteq Q_\calA$, for $q\in S$,
  such that for every play of $\sigma$ on some $u\in L_a$ with final history string $\alpha$,
  it holds $\delta_\calA^*(s_\calA,\alpha)\in Q_{\calA,q} \Leftrightarrow \delta^*(p,\natural\alpha)=q$.\\ 
Furthermore, there is no proper
  prefix $\beta$ of $\alpha$ for which $\delta_\calA^*(s_\calA,\beta)\in
  Q_{\calA,r}$, for any $r$.
\end{itemize}

\begin{proposition}\label{prop:OnlineToGame}
Let $G=(\Sigma,R,T)$ be a prefix-free game and $E$ a set of effect
triples such that for each non-trivial $t\in E$, there exists a $t$-inducing
strategy automaton $\calA_t$. Then there is a strategy automaton $\calA$ for
$G$ such that, for each strategy
$\rho$ for $\online[\calN_E]$, $ W_{\calN_E}(\rho) \lesl W_G(\sigma_\calA)$.
\end{proposition}
\full{
\begin{proof}
Let $\calD$ be the DFA guaranteed by Lemma~\ref{lem:OnlineNFA}, giving the weakly dominant strategy $\rho_\calD$ for
\online[\calN_E].
	It suffices to construct $\calA$ such that $ W_{\calN_E}(\rho_\calD) \subseteq W_G(\sigma_\calA)$.
	
	The strategy automaton $\mathcal A$ has a non-accepting state $(p,S)$ for
        each $ S\subseteq Q$ and $p\in S$, further states of the form
        $(S,r)$, where $r$ is a state of some
        $\calA_t$, and $(S,r)$ is accepting if and only if $r$ is
        accepting in $\calA_t$, and an accepting state \Call. The
        initial state is $(s,\{s\})$. We construct $\calA$ such that,
        if $\alpha$ is the final history string of a play of
        $\sigma_\calA$ on $w$, then
        $\delta_\calA^*(s_\calA,\alpha)=(p,S)$, where
        $S=\rho_\calD(w)$, $p=\delta^*(s,\natural\alpha)$, and $p\in
        S$. 
Thus, whenever a subplay on a prefix of the input word is completed, $\calA$
knows the state $p$ that has been reached, and the state $S$ that $\rho$
has reached in $\calN$. 

To this end, $\delta_\calA$ is defined as follows. Let
$(p,S)$ be the current state and $a$ the next symbol and let
$S'=\delta_\calD(S,a)$. 
If  $\delta(p,a)\in S'$ then from $p$ a state in $S'$ can be
reached by playing \Read on $a$ and thus
$\delta_\calA((p,S),a)=(\delta(p,a),S')$. Otherwise,
$\delta_\calA((p,S),a)=\Call$. In the latter case, we can choose some $t=(p,a,S'')\in E$ with $S''\subseteq S'$, corresponding to a $t$-inducing automaton $\calA_t$. 
To continue the play after playing \Call on $a$, let $\delta_\calA((p,S),\call{a})=(S',s_t)$, where $s_t$ is the initial state of $\calA_t$. That is,
a subrun is started in which $\calA_t$ is simulated on the replacement word of $a$, and $S'$
is laid aside. 

To carry out the simulation of $\calA_t$, we simply define
$\delta_\calA((S',x),b)=(S',\delta_t(x,b))$ for every state $x$ of $\calA_t$ and every $b\in\hist{\Sigma}$, unless $\delta_t(x,b)\in Q_{\calA_t,q}$ for some $q$. If $\delta_t(x,b)\in Q_{\calA_t,q}$, then
$\delta_\calA((S',x),b)=(q,S')$, ending the
subgame. 

By induction on the length of $w$, it is straightforward to show that $\calA$ meets
the above specification. In particular, whenever $\rho_\calD$ wins for
a word $w$ in $\online[\calN_E]$, $\calA$ yields a strategy that ends
in a state of the form $(p,S)$ with $S\subseteq F$ and $p\in S$ and it
holds $\delta^*(s,\natural\alpha)=p\in F$, for the final history
string $\alpha$. Therefore, $W_{\calN_E}(\rho_\calD) \subseteq W_G(\sigma)$ indeed
holds.
\end{proof}
} % \full

\full{
	The next lemma gives a crucial property of prefix-free games. Intuitively, it says that Juliet always knows the maximal prefix $u$ of the input word that is completely processed, the history string $\alpha$ corresponding to the subgame on $u$, and the next symbol $b$ of the input word that is currently processing.
\begin{lemma}\label{lem:pfree}
	Let $(\gamma,c,v)$ be a configuration in a play $\Pi(\sigma,\tau,w)$ of a prefix-free game $G=(\Sigma,R,T)$. Then $\gamma$ and $c$ uniquely determine $u\in\Sigma^*$, $b\in\Sigma$, and a prefix $\alpha$ of $\gamma$ such that $ub$ is a prefix of $w$, $\alpha$ is the final history string of $\Pi(\sigma,\tau,u)$, and $u$ has maximal length.
\end{lemma}
\begin{proof}
	Let $(K_0,K_1,\dots,K_n)$ be the prefix of $\Pi(\sigma,\tau,w)$ with
	$K_i=(\gamma_i,c_i,v_i)$, for each $i$, and
	$(\gamma_n,c_n)=(\gamma,c)$. We show by backwards induction that for
	$i=n,n-1,\dots,0$, there are unique $u_i,b_i,\alpha_i$ such that
	$\gamma_i\alpha_i$ is a prefix of $\gamma$, $u_ib_i$ is a prefix of
	$c_iv_i$, $\alpha_i$ is the final history string of
	$\Pi(\sigma^{\gamma_i},\tau^{\gamma_i},u_i)$, and $u_i$ has maximal
	length. In a nutshell, from $K_i$ onwards, the string to be
	played on is $c_iv_i$. The string $u_ib_i$ is the prefix of  $c_iv_i$
	that has been ``touched''  by the play from $K_i$ to
	$K_n$. The subplay on $u_i$ yielded the additional history string
	$\alpha_i$. The subplay on $b_i$ is ongoing at $K_n$.
	The lemma then follows by choosing
	$(u,b,\alpha)=(u_0,b_0,\alpha_0)$.
	
	The base case is $u_n=\alpha_n=\epsilon$ and $b_n=c_n$. For the
	induction step, either $\gamma_i=\gamma_{i-1}c_{i-1}$ or
	$\gamma_i=\gamma_{i-1}\call{c_{i-1}}$. In the former case,
	$u_{i-1}=c_{i-1}u_i$, $b_{i-1}=b_i$,
	$\alpha_{i-1}=c_{i-1}\alpha_i$. Otherwise, we consider two
	sub-cases. If $u_i=xx'$ for some $x\in L_{c_{i-1}}$ and $x'\in\Sigma^*$
	(here, $x$ is uniquely determined thanks to prefix-freeness of $ L_{c_{i-1}}$), then
	$\tau(\gamma_{i-1},c_{i-1})=x$. Thus, we have
	$K_i=(\gamma_{i-1}\call{c_{i-1}},\tau(\gamma_{i-1},c_{i-1})x'b_iz)$ for
	some $z$ and hence $K_{i-1}=(\gamma_{i-1},c_{i-1}x'b_iz)$. Then
	$u_{i-1}=c_{i-1}x'$, $b_{i-1}=b_i$,
	$\alpha_{i-1}=\call{c_{i-1}}\alpha_i$. Otherwise, the subgame for
	$c_{i-1}$ has not ended in $K_n$, and thus $u_{i-1}=\alpha_{i-1}=\epsilon$, $b_{i-1}=c_{i-1}$. It
	is easy to verify in all cases that $(u_{i-1},b_{i-1},\alpha_{i-1})$
	is the unique  solution for the required properties.
\end{proof}
}

A crucial ingredient is the following proposition, which will allow us to restrict our attention to strategies of finite depth. An almost identical proof can also be used to show that prefix-free games have the bounded depth property, and we could use this and Theorem~\ref{theo:right}~(a) to deduce immediately that they have weakly dominant strategies. However, we are aiming for the stronger result that they have \emph{regular} weakly dominant strategies.

\begin{proposition}\label{prop:prefixfreeBDP}
	In prefix-free games, each effect triple of a terminating strategy is also an effect triple of a strategy of bounded depth.
\end{proposition}
\short{
\begin{proof}[Proof idea]
	Let $E$ denote the set effect triples of bounded depth strategies. Consider an effect triple $(p,a,S)\notin E$ of a strategy $\hat\sigma$. If the effect triples of $\hat\sigma$'s substrategies on the symbols of $a$'s replacement word were all in $E$, then these substrategies could be replaced so as to obtain a bounded depth strategy $\sigma$ for $(p,a,S)$, contradicting $(p,a,S)\notin E$. Thus, for each $(p,a,S)\notin E$, Romeo can force a configuration in the play against $\hat\sigma$ on $a$ where the effect triple of the substrategy is again not in $E$. But Romeo can do this repeatedly, so $\hat\sigma$ is not terminating.
\end{proof}
}
\full{
	\begin{proof}
		Let $E$ denote the set effect triples of all strategies of bounded depth. For $t\in E$ fix such a bounded depth strategy $\sigma_t$ with effect triple $t$. Since $E\subseteq Q\times\Sigma\times \mathcal P(Q)$ is finite, there exists $B$ such that for each $t\in E$, the depth of plays of $\sigma_t$ is at most $B$.

We claim that if $(p,a,S)\notin E$ is an effect triple of strategy $\hat\sigma$, then there is a Romeo strategy $\tau$ such that $\Pi(\hat\sigma,\tau,a)$ contains a configuration $(\call a\alpha,b,w)$ with $(q,b,S')\notin E$ for $q=\delta^*(p,\natural\alpha)$ and $S'=\states{q,b}{\hat\sigma^{\call a\alpha}}$. But this means that Romeo can repeatedly force a configuration where $\hat\sigma$'s substrategy on the current symbol has an effect triple outside of $E$, and hence $\hat\sigma$ would not be terminating. Thus, each effect triple of a terminating strategy is in $E$.

It remains to show the claim. Let $(p,a,S)\notin E$ be an effect triple of a strategy $\hat\sigma$. Suppose for the sake of contradiction that for each configuration of the form $(\call a\alpha,b,w)$ occurring in some play of $\hat\sigma$ on $a$ it holds that $(q,b,S')\in E$, where $q=\delta^*(p,\natural\alpha)$ and $S'=\states{q,b}{\hat\sigma^{\call a\alpha}}$. We claim that in this case, $(p,a,S)$ is an effect triple of a strategy of depth at most $B+1$, in contradiction to $(p,a,S)\notin E$.

We first define a strategy $\sigma$ that plays on input words from $L_a$. We define $\sigma(\gamma,c)$ for $(\gamma,c)\in\hist\Sigma^*\times\Sigma_f$ such that a configuration of the form $(\gamma,c,v)$ occurs in some play of $\sigma$ on an input word from $L_a$ (given the definition of $\sigma$ for previous configurations). Consider $u,b,\alpha$ determined by $(\gamma,c)$ as per Lemma~\ref{lem:pfree}. The definition of $\sigma$ will ensure by induction on the length of $u$ that $\states{p,u}{\sigma}\subseteq\states{p,u}{\hat\sigma^{\call a}}$. Clearly this holds for $u=\epsilon$.

Let $q=\delta^*(p,\natural\alpha)$. Since $q\in\states{p,u}{\sigma}\subseteq\states{p,u}{\hat\sigma^{\call a}}$, there is a history string $\alpha'$ that is the final history string of a play of $\hat\sigma^{\call a}$ on $u$ such that $q=\delta^*(p,\natural\alpha')$. Choose such $\alpha'$ lexicographically minimal. Since $\hat\sigma(\epsilon,a)=\Call$ (because $(p,a,S)\notin E$) and $ubw\in L_a$ for some $w$, the configuration $(\call a\alpha',b,w)$ occurs in some play of $\hat\sigma$ on $a$. Let $S'=\states{q,b}{\sigma^{\call a\alpha'}}$. By assumption, $(q,b,S')\in E$. For each non-final configuration $(\beta,c,z)$ occurring in some play of $\sigma_{q,b,S'}$, let $\sigma(\alpha\beta,c)=\sigma_{q,b,S'}(\beta,c)$. Note that $\gamma=\alpha\beta$ for some such $\beta$. This defines $\sigma$ until $ub$ is completely processed. Since $S'\subseteq \states{p,ub}{\sigma^{\call a}}$, we have $\states{p,ub}{\sigma}\subseteq\states{p,ub}{\hat\sigma^{\call a}}$, which completes the induction step.

Consider now the strategy $\sigma'$ given by $\sigma'(\epsilon,a)=\Call$, $\sigma'(\call a\gamma,c)=\sigma(\gamma,c)$ if the latter is defined, and defined as $\Read$ everywhere else. Since $$\states{p,u}{\sigma}\subseteq\states{p,u}{\hat\sigma^{\call a}}\subseteq\states{p,a}{\hat\sigma}=S$$ for all $u\in L_a$, it follows $\states{p,a}{\sigma'}\subseteq S$ and hence $(p,a,S)$ is an effect triple of $\sigma'$. But the depth of $\sigma'$ is at most $B+1$, contradicting $(p,a,S)\notin E$.
\end{proof}
} %\full

The finite depth allows us to construct $t$-inducing strategy automata by induction on the depth of a strategy with effect triple $t$.

\begin{proposition}\label{prop:sse}
 Let $G=(\Sigma,R,T)$ be a prefix-free game and $(p,a,S)$ a non-trivial effect triple of some terminating strategy $\sigma$. Then there exists a $(p,a,S)$-inducing strategy
 automaton.
\end{proposition}
\full{
\begin{proof}
Due to Proposition~\ref{prop:prefixfreeBDP}, we can assume that the depth of $\sigma$ is bounded by some finite $k$. We proceed by induction on $k$.

Since $(p,a,S)$ is non-trivial, $\sigma(\epsilon,a)=\Call$.  Let
$\sigma'$ be the strategy defined by $\sigma'(\alpha,b)=\sigma(\call{a}\alpha,b)$ if a configuration $(\call a\alpha,b,u)$ can occur in a play of $\sigma$ on $a$, for some $u$,
and $\sigma'(\alpha,b)=\Read$ otherwise. Note that the depth of $\sigma'$ is
at most $k-1$. Moreover, for every $u\in L_a$, it holds
$\states{p}{u,\sigma'}\subseteq S$, since $\states{p}{a,\sigma}\subseteq S$.

Let $T'$ be the target automaton $(Q\times
Q_a,\Sigma,\delta\times\delta_a,(p,s_a),S\times \{f_a\})$, that is, the
product automaton of $T^p$ and $\calA_a$ whose accepting states are
chosen to require $T^p$ to end in
$S$ and $\calA_a$ to end in $f_a$. Let $G'$ be the game
$(\Sigma,R,T')$. We consider $\sigma'$ as a strategy for
$G'$. Since $\states{p}{u,\sigma'}\subseteq S$ for each $u\in L_a$, we have $L_a\subseteq W_{G'}(\sigma')$.

Let $E$ be the set of effect triples $(q,b,S')$ of strategies for $G'$ with
depth at most $k-1$. For each non-trivial $t\in E$, a
$t$-inducing strategy automaton exists by induction.
Let
$\calA$ be the strategy automaton for the set $E$, defined in the proof of Proposition~\ref{prop:OnlineToGame}.

By Proposition~\ref{prop:GameToOnline},  there is a strategy
$\rho$ for $\online[\calN_E]$ such that $W_{G'}(\sigma')=
W_{\calN_E}(\rho)$. By Proposition~\ref{prop:OnlineToGame},
$W_{\calN_E}(\rho)\lesl W_{G'}(\sigma_\calA)$. Altogether, we have
$L_a\lesl W_{G'}(\sigma_\calA)$ and, by definition of $G'$,
$W_{G'}(\sigma_\calA)\subseteq L_a$,  hence $W_{G'}(\sigma_\calA)= L_a$. Therefore, $\calA$ is $(p,a,S)$-inducing if we choose $Q_q=\{((q,f_a),S')\mid S'\ni(q,f_a)\}$ for every $q\in S$.
  \end{proof}
} % \full
Now we are ready to prove Theorem~\ref{theo:right}~(b).

\begin{proofof}{Theorem~\ref{theo:right}~(b)}
   Let $G=(\Sigma,R,T)$ be prefix-free. Let $E$ be the set of effect triples of terminating strategies. By
   Proposition~\ref{prop:sse} there is a $t$-inducing automaton for
   each non-trivial $t\in E$. Let $\sigma_\calA$ be
the regular strategy as guaranteed by Proposition~\ref{prop:OnlineToGame}. We
show that $\sigma_\calA$ is weakly dominant.

To this end, let $\sigma$ be any terminating strategy for $G$. Since $E(\sigma)\subseteq
E$, Proposition~\ref{prop:GameToOnline} guarantees a strategy $\rho$
for $\online[\calN_E]$ such that $W_G(\sigma)=
W_{\calN_E}(\rho)$. By Proposition~\ref{prop:OnlineToGame}, $
W_{\calN_E}(\rho) \lesl W_G(\sigma_\calA)$, and therefore, altogether
$W_G(\sigma)  \lesl W_G(\sigma_\calA)$ as required.
\end{proofof}

\subsubsection{Games with finite target language}\label{sec:ex-nr-finite}
\full{
\begin{proofof}{Theorem~\ref{theo:right}~(c)}
Let $G = (\Sigma,R,T)$ with $T=(Q,\Sigma,\delta,s,F)$ such that $L(T)$ is finite, and let $\sigma$ be a dominant strategy for $G$. Thanks to Lemma~\ref{lem:alwaysTerm} we can assume that all plays of $\sigma$ are finite.

We say that $\sigma$ has a \term{$(q,a)$-conflict}, for $q\in Q$ and $a\in\Sigma_f$, if there are words $w_1,w_2\in W(\sigma)$ and strategies $\tau_1,\tau_2$ of Romeo such that $\Pi(\sigma,\tau_1,w_1)$ reaches some configuration $(\alpha_1,a,u_1)$ with $\delta^*(s,\natural\alpha_1)=q$ and $\sigma(\alpha_1,a)=\Read$,  and 
$\Pi(\sigma,\tau_2,w_2)$ reaches some configuration $(\alpha_2,a,u_2)$
with $\delta^*(s,\natural\alpha_2)=q$ and
$\sigma(\alpha_2,a)=\Call$. By $C(\sigma)$ we denote the set of pairs
$(q,a)$, for which $\sigma$ has a $(q,a)$-conflict. 

If $C(\sigma)=\emptyset$, then $\sigma$ is already almost strongly
regular. The only exceptions might be configurations that are not
reached by plays on strings from  $W(\sigma)$ and
we can easily change $\sigma$ on those configurations to make it
entirely  strongly regular without affecting $W(\sigma)$. 

Thus we only need to show that a dominant strategy $\sigma$ with some
conflicts can be transformed into a dominant strategy $\sigma'$ with
$C(\sigma')\subsetneq C(\sigma)$. By repeating this process at most
$|Q||\Sigma_f|$ times, we reach a strategy without conflicts. 

To this end, let $\sigma$ have a $(q,a)$-conflict, for some $q$ and
$a$, with $\tau_1,\tau_2,\alpha_1,\alpha_2,u_1,u_2$ as above such that
$\alpha_1$ has minimal possible length. Let $v_1=\natural\alpha_1$.

We first claim that $\sigma$ only plays Read moves on the input word $v_1$. 
As a first step towards this claim, we show $v_1au_1\in
W(\sigma)$. Indeed, the strategy $\widetilde{\sigma}$ that plays Read
on $v_1$ and is defined by
$\widetilde{\sigma}(v_1\beta,x)=\sigma(\alpha_1\beta,x)$ afterwards
(and arbitrarily on other configurations) clearly wins on $v_1au_1$,
since $\sigma$ wins from $(\alpha_1,a,u_1)$ on. As $\sigma$ is
dominant, therefore $v_1au_1\in W(\sigma)$ holds. 
As $\sigma$ plays Read on all $\Sigma$-symbols of $\alpha_1$, the assumption that $\sigma$ plays Call on some prefix of $v_1$ would yield a smaller instance of a conflict (with respect to the length of $\alpha_1$), and therefore, the claim that   $\sigma$ only plays Read moves on the prefix of $v_1$ is shown.

We say that $\sigma$ is \emph{$(q,a)$-conflict-free after $v_1$} if either $\sigma(v_1,a)=\Read$ and there is no $(q,a)$-conflicting
\Call-move in plays of $\sigma$ on any word $v_1au\in W(\sigma)$, for some $u$, or vice versa with \Read and \Call reversed. In both cases, a strategy without any $(q,a)$-conflicts can be obtained by ``copying'' the substrategy of the subgame starting from $(v_1,a,u)$ to any conflicting configuration.

However, we need to show how to guarantee $(q,a)$-conflict-freeness after $v_1$.
Actually, thanks to the finiteness of $L(T)$
there can be no such conflict if $\sigma(v_1,a)=\Read$, since $T$ can not return to state $q$ after taking a
transition from $q$. But  in the case where
$\sigma(v_1,a)=\Call$, there could be a conflicting $\Read$-move at
some configuration $(\alpha,a,u')$ after $(v_1,a,u)$. However, again thanks to the finiteness of
$L(T)$, there can not be any \Read moves between $(v_1,a,u)$ and
$(\alpha,a,u')$ and therefore, in particular, $u$ must be a suffix of $u'$. 
In this case we show that actually $u'=u$ must hold. By copying the strategy of $\sigma$ of the
subgame on  $(\alpha,a,u)$ to the subgame on $(v_1,a,u)$, a strategy that matches the first
case above can be obtained. 

In all cases, it can be concluded from the dominance of $\sigma$ that the new strategies obtained by the described modifications are
again dominant.

We now turn to a more detailed description of the proof.

We first show how to get rid of $(q,a)$-conflicts after $v_1$. Let us thus assume that $\sigma(v_1,a)=\Call$, $v_1au\in W(\sigma)$, and for some strategy $\tau$ of Romeo,
$\Pi(\sigma,\tau,v_1au)$ contains a configuration $(\alpha,a,u')$ with
$\delta^*(s,\natural\alpha)=q$ and $\sigma(\alpha,a)=\Read$, for some $u'$. As already argued above, there can
not be any \Read moves between  $(v_1,a,u)$ and $(\alpha,a,u')$ and
therefore $u$ must be a suffix of $u'$. 
We claim that $u'=u$.

Towards a contradiction, let us assume that $u'=u''u$ for some $u''\not=\epsilon$. 
We show by induction on $k$ that
$v_1a(u'')^ku\in W(\sigma)$ for each $k$. This is true for $k=0$ by
assumption. Suppose now that $v_1a(u'')^ku\in
W(\sigma)$. Similarly as
before, the play $\Pi(\sigma,\tau,v_1a(u'')^ku)$ reaches configuration
$(v_1,a,(u'')^ku)$ and later on $(\alpha,a,(u'')^{k+1}u)$. From this
configuration on, $\sigma$ is winning. 
But then the strategy that plays from $(v_1,a,(u'')^{k+1}u)$ on like
$\sigma$ plays from $(\alpha,a,(u'')^{k+1}u)$ on wins on the string
$v_1a(u'')^{k+1}u$. By dominance of $\sigma$, it follows that $v_1a(u'')^{k+1}u\in W(\sigma)$, completing the inductive step. 

However, now we can choose $k$  larger than the length $\ell$ of the longest string
in $L(T)$ to get the desired
contradiction: $v_1a(u'')^{k}u$ has more than $\ell$ symbols and
each of them contributes at least one symbol from $\Sigma$ to the
final history string, which therefore can not yield a win for Juliet.
 Altogether, we have established that $u'=u$.

We now define
$\widetilde{\sigma}$ by  
\[
\widetilde{\sigma}(\beta,x)=
\begin{cases}
  \sigma(\alpha\beta',x) & \text{if $\beta=v_1\beta'$, and $\beta'x$ begins with $a$,}\\
\sigma(\beta,x) & \text{otherwise.}
\end{cases}
\]
That $W(\sigma)\subseteq W(\widetilde{\sigma})$ is easy to show: if the
difference between $\sigma$ and $\widetilde{\sigma}$ actually matters for some
word $w\in W(\sigma)$, then $w$ must be of the form $v_1aw'$ and the
first configuration in which $\sigma$ and $\widetilde{\sigma}$ differ is
$(v_1,a,w')$. However, since Romeo can enforce the configuration
$(\alpha,a,w')$, $\sigma$ is winning from that configuration on. Since
$\widetilde{\sigma}$ plays from $(v_1,a,w')$ on like $\sigma$ plays from
$(\alpha,a,w')$ on, we can conclude that $\widetilde{\sigma}$ is winning from
$(v_1,a,w')$, as desired. 

Therefore, in the following we can assume that $\sigma$ is $(q,a)$-conflict-free after $v_1$.

Now we are ready to get rid of the $(q,a)$-conflict altogether. We define a strategy $\sigma'$ that copies the substrategy of $\sigma$ from configurations of the form $(v_1,a,u)$ onwards to any conflicting configuration. Formally,
\[
\sigma'(\beta,x)=
\begin{cases}
  \sigma(v_1\beta'',x) & \text{
    \begin{minipage}[t]{8cm}
      if $\beta=\beta'\beta''$,
      $\delta^*(s,\natural\beta')=q$,
      $\beta''x$ begins
      with $a$ or $\call{a}$, and $\sigma(\beta',a)\ne \sigma(v_1,a)$, where $\beta'$ is minimal
    \end{minipage}
}\\
\sigma(\beta,x) & \text{otherwise.}
\end{cases}
\]

Since the sub-plays of $\sigma$ from $(v_1,a,u)$ on do not contain any configurations that can cause a $(q,a)$-conflict, and since clearly no new conflicts are induced, we can conclude that $C(\sigma')\subsetneq C(\sigma)$.

To show that $\sigma'$ is again dominant it suffices to show
$W(\sigma)\subseteq W(\sigma')$. If the
difference between $\sigma$ and $\sigma'$ matters for some
word $w\in W(\sigma)$, then the
first configuration in which $\sigma$ and $\sigma'$ differ is of the form
$(\beta',a,u)$, with
$\beta'$ as in the definition of $\sigma'$ and $u\in\Sigma^*$. Since $\sigma$ is winning from $(\beta',a,u)$ on, a
strategy that plays \Read on $v_1$ and plays from
$(v_1,a,u)$ on like $\sigma$ from $(\beta',a,u)$ on wins on the word
$v_1au$. Since $\sigma$ is dominant, $\sigma$ wins on $v_1au$ as well. Since
$\sigma'$ plays from $(\beta',a,u)$ on like $\sigma$ from
$(v_1,a,u)$, it also wins, as desired.
\end{proofof}
}%\full

\subsubsection{Non-recursive games}\label{sec:ex-nr-nonrecursive}
\short{
	\begin{proof}[Proof idea of Theorem~\ref{theo:right}~(d)]
		Starting with some dominant strategy $\sigma_1$, we inductively construct a
		sequence $(\sigma_1,\sigma_2,\dots)$ of dominant strategies that is
		either finite and ends with a strategy that has no $(q,a)$-conflict, or is infinite and converges to such a strategy. In the convergent case, since each strategy in a non-recursive game is terminating, the limit strategy is also dominant by Lemma~\ref{lemma:conv}. To construct a strategy $\sigma_{k+1}$, the idea is to modify $\sigma_k$ so as to move the earliest $(q,a)$-conflict to a later time in the future.
	\end{proof}
}
\full{
\begin{proofof}{Theorem~\ref{theo:right}~(d)}
Let $G=(\Sigma,R,T)$ with $T=(Q,\Sigma,\delta,s,F)$. Starting with some dominant strategy $\sigma_1$, we inductively construct a
sequence $(\sigma_1,\sigma_2,\dots)$ of dominant strategies that is
either finite and ends with a strategy that has no $(q,a)$-conflict, or is infinite and converges to such a strategy. To construct a strategy $\sigma_{k+1}$, the idea is to modify $\sigma_k$ so as to move the earliest $(q,a)$-conflict to a later time in the future.

We describe next how $\sigma_{k+1}$ is defined from $\sigma_k$. We distinguish two cases.

The first case is that for all $q\in Q$, $a\in\Sigma_f$ and all
configurations $(\alpha_1,a,u_1)$ and $(\alpha_2,a,u_2)$ with
$q=\delta^*(s,\natural\alpha_1)=\delta^*(s,\natural\alpha_2)$ that
occur in plays of $\sigma_k$ on strings from $W(\sigma_k)$, it holds that
$\sigma_k(\alpha_1,a)=\sigma_k(\alpha_2,a)$. In this case
$\sigma_k$ is the last strategy of the sequence.

In the other case,  let $w_1,w_2\in W(\sigma)$ and strategies $\tau_1,\tau_2$ of Romeo such that $\Pi(\sigma_k,\tau_1,w_1)$ reaches some configuration $(\alpha_1,a,u_1)$ with $\delta^*(s,\natural\alpha_1)=q$,
$\Pi(\sigma_k,\tau_2,w_2)$ reaches some configuration $(\alpha_2,a,u_2)$
with $\delta^*(s,\natural\alpha_2)=q$ and
$\sigma_k(\alpha_1,a)\not=\sigma_k(\alpha_2,a)$, such that $|\alpha_2|\ge|\alpha_1|$ and $(|\alpha_2|,|\alpha_1|)$ is lexicographically minimal.

As in the proof of
Theorem~\ref{theo:right}~(c), it can be shown that $\sigma_k$ only
plays \Read moves on $\natural\alpha_1$ and that it wins on
$\natural\alpha_1au_1$. By minimality of $(|\alpha_2|,|\alpha_1|)$ with $|\alpha_2|\ge|\alpha_1|$, we can conclude that actually $\alpha_1=\natural\alpha_1$.

Now we are ready to define $\sigma_{k+1}$. It plays like $\sigma_k$
except that if it reaches a configuration with
$T$-state $q$ and current symbol $a$ after precisely $|\alpha_2|$ moves, then it continues to play like
$\sigma_k$ for the history string $\alpha_1$. Formally,
\[
\sigma_{k+1}(\beta,x)=
\begin{cases}
  \sigma_k(\alpha_1\beta'',x) & \text{
    \begin{minipage}[t]{7.8cm}
      if $\beta=\beta'\beta''$, $|\beta'|=|\alpha_2|$, $\delta^*(s,\natural\beta')=q$, and
      $\beta''x$ begins with $a$ or $\call{a}$,
    \end{minipage}
}\\
\sigma_k(\beta,x) & \text{otherwise.}
\end{cases}
\]

We need to show that $\sigma_{k+1}$ is dominant. Towards a
contradiction, suppose it is 
not.  
Since $\sigma_k$ is dominant by the induction hypothesis, there exist $w\in
W(\sigma_k)$ and a Romeo strategy $\tau$ such
that Juliet does not win $\Pi(\sigma_{k+1},\tau,w)$. Thus, the plays of
$\sigma_{k+1}$ and $\sigma_k$ against $\tau$ on $w$ reach a
configuration $(\beta,av)$ with $\beta\in\hs$ and
$v\in\Sigma^*$ such that $\delta^*(s,\natural\beta)=q$. Furthermore,
$\sigma_k$ wins from $(\beta,av)$ on, but $\sigma_{k+1}$ does not, and
therefore $\sigma_k$
does not win on $\alpha_1av$. Since $\sigma_k$ is dominant, no
strategy wins on $\alpha_1av$. However, consider the strategy $\widetilde{\sigma}$
defined by 
\begin{align*}
\widetilde{\sigma}(\alpha,b)=
\begin{cases}
\sigma_k(\alpha,b),&\text{if $\alpha_1$ is not a prefix of $\alpha$,}\\
\sigma_k(\beta\alpha',b),&\text{if $\alpha=\alpha_1\alpha'$ for some $\alpha'\in\hs$.}
\end{cases}
\end{align*}
 On the input word $\alpha_1av$, $\widetilde{\sigma}$ initially plays only Read moves on the prefix $\alpha_1$. Upon reaching the configuration $(\alpha_1,av)$, it starts playing like $\sigma_k$ would play after reaching the configuration $(\beta,av)$.
Since $\sigma_k$ wins from $(\beta,av)$ on, $\widetilde{\sigma}$ wins on $\alpha_1av$, the desired contradiction. 

Altogether, $(\sigma_1,\sigma_2,\dots)$ is indeed a sequence of dominant
strategies. If it is finite, then its last strategy $\sigma_k$ has no two conflicting
configurations that can occur in plays on words from $W(\sigma_k)$. Thus, it can be turned into a strongly regular dominant strategy by adapting it only for configurations that cannot occur in plays on words where it does not win.
 
If $(\sigma_1,\sigma_2,\dots)$ is infinite, it converges towards a strategy
$\sigma'$, because, for each $n$, $\alpha_2$ can be of length $n$  at most $|Q||\Sigma_f|$
times in the above construction. Since $G$ is non-recursive, every play terminates, and
therefore Lemma~\ref{lemma:conv} guarantees that $\sigma'$ is
dominant as well. Finally, $\sigma'$ can be transformed into a
strongly regular strategy just as in the finite case.
\end{proofof}
We note that the proof uses non-recursiveness only to guarantee that
$\sigma'$ has the same winning set as every $\sigma_k$. 

\subsubsection{Unary games}\label{sec:unary}

To complete the proof of Theorem~\ref{theo:right} we show in this subsection that every unary game has an undominated strategy. Before we turn to this proof, we first present some observations and a helpful lemma. 
Let in the following always $\Sigma=\{a\}$ and 

We first observe that, since the DFA $T$ for the unary target language is minimal, we can assume that it is of the following form: $Q=\{0, 1, \dots, m-1\}$, for some $m$, with $1$ as initial state; the transitions are precisely the edges $(i-1, i)$ for $i=1,\dots,m-1$ and the edge $(m-1,j)$ for some state $j\in Q$. 

We can assume that $a\not\in R_a$, since otherwise always playing Read is a dominant strategy. Indeed, playing Call cannot yield an advantage in this case, since Romeo can choose to never change the current word. Thus, we can restrict our attention to the case where the replacement language contains only words of length at least $2$.

We call a strategy $\sigma$ \term{almost undominated} if there are only finitely many words on which $\sigma$ does not win but some strategy dominating $\sigma$ does win.

\begin{lemma}\label{lem:almost}
  If a game has an almost undominated strategy then it also has an undominated strategy.
\end{lemma}
\begin{proof}
  Let $\sigma_0$ be an almost undominated strategy, let $W_0$ be the finite set of words  on which $\sigma_0$ does not win but some strategy dominating $\sigma_0$ does win, and let $n=|W_0|$. For each $i>0$, we define $\sigma_i$ as $\sigma_{i-1}$ if  $\sigma_{i-1}$ is undominated or as some strategy that dominates $\sigma_{i-1}$ and wins on the smallest word in $W_{i-1}$. We define $W_i$ as the subset of $W_{i-1}$ of words on which $\sigma_i$ does not win but some other strategy dominating $\sigma$ does win. Clearly, $W_n$ is empty and thus $\sigma_n$ is undominated.
\end{proof}

\begin{proofof}{Theorem~\ref{theo:right}~(e)}
Thanks to Lemma~\ref{lem:almost} it suffices to show that any unary game has an almost undominated strategy.

	We may assume that the target language $L(T)$ is $m$-periodic,
        meaning that $\delta(m-1,a)=0$. The theorem for the general
        case where $\delta(m-1,a)>0$ is then implied as follows: the
        target language of such a general automaton differs from some
        $(m-j)$-periodic language by at most the $j$ shortest
        words. Hence, if a strategy is almost undominated for the
        associated $(m-j)$-periodic language, then it is also almost
        undominated for the actual target language. We can also assume
        that $|R_a|>1$, since otherwise $G$ is prefix-free.

In the following, we identify a string $a^\ell$ with the number $\ell$.	 Let $n, k\ge 1$ be such that $1+n$ and $1+n+k$ are the shortest words in $R_a$. 

	Since every Call extends the remaining string by at least $n$ symbols, for each strategy there is a strategy with (at least) the same winning set, which always plays at least $n$ Reads after each Call. We denote such a combined move by CR$^{n}$. We therefore assume in the following that all strategies of Juliet that we consider are of this normal form and we only consider configurations and history strings where every CR$^n$ is complete. Moreover, whenever the successor state of the $T$-state of the current configuration is accepting, Juliet plays Read (because if this Read reaches the end of the word, then Juliet wins, and otherwise she does not hurt her chances to win by delaying Calls until later).

With each configuration $K$ we associate a profile $(c,r,\ell)$
consisting of the number $c$ of CR$^n$s that were played so far, the
number $r$ of additional Reads played so far (outside CR$^n$ moves)\footnote{Thus the total number of Reads before $K$ is $cn+r$.} and the length $\ell$ of the remaining string.
We note that history strings for games with unary alphabet are very uniform. Indeed, for arbitrary strategies $\sigma$ of Juliet, strategies $\tau,\tau'$ of Romeo, strings $w,w'$ and $i\ge 0$ it holds: if the plays $\Pi(\sigma,\tau,w)$ and $\Pi(\sigma,\tau',w')$ last for at least $i$ moves, then their history strings after $i$ moves are identical. In particular, for every fixed $\sigma$, each profile determines uniquely its underlying configuration.

By $\sigma_0$ we denote Juliet's strategy that always reads, and by $\tau_0$ and $\tau_1$ we denote the strategies of Romeo that always play $1+n$ and always play $1+n+k$, respectively.
We say that a strategy $\sigma$ is \emph{patient} if there is a string $w$ such that $\Pi(\sigma,\tau_0,w)$ has at least $m$ Call-moves.

Let $L$ be the union of $L(T)$ and the set of all strings of the form $w+ik$ with $w\in L(T)$, $w\ge m$, and $i\in \{1,\ldots,m-1\}$. We note that $L$ is ultimately $k$-periodic. More precisely, for all words $w$ of length at least $(k+1)|Q|$ it holds that $w\in L$ if and only if $w+k\in L$.

The existence of an almost undominated strategy can be concluded from the following three claims:
	\begin{enumerate}[(a)]
        \item If $\sigma_0$ is not undominated then there is a patient strategy that dominates $\sigma_0$. \label{it:unaryPatient}
		\item If there is a patient strategy that dominates $\sigma_0$, then there is a strategy that wins on $L$. \label{it:unaryInducedGood}
		\item Every strategy that wins on $L$ is almost undominated. \label{it:unaryNoneBetter}
	\end{enumerate}
	
	\ \\
	\noindent \textit{Proof of (\ref{it:unaryPatient}):}
	
	\noindent The proof is by contradiction. We assume that $\sigma_0$ is not undominated but that every strategy $\sigma$ that dominates it, uses at most $C$ Call-moves, for some $C<m$.
Let $\sigma$ be a strategy dominating $\sigma_0$ that has maximally $C$ Call moves.
Let $w$ be sufficiently large and $K$ be a configuration in the play $\Pi(\sigma,\tau_0,w)$ with $T$-state $0$ with $C$ Call moves in its history string $\alpha$. Let $\sigma'$ be the strategy that plays $\sigma'(\beta,a)=\sigma(\beta,a)$ for strict prefixes $\beta$ of $\alpha$ and $\sigma'(\alpha\beta,a)=\sigma(\beta,a)$, otherwise.
It is easy to see that $\sigma'$ has more than $C$ Call moves and dominates $\sigma_0$, the desired contradiction.
	
	\ \\
	\noindent \textit{Proof of (\ref{it:unaryInducedGood}):}
	
	\noindent 

Let $\sigma$ be a patient strategy that dominates $\sigma_0$.
Let $K$ be some configuration with profile $(m,r,\ell)$ and for which Read was played in the last move. We claim that $\sigma$ wins from $K$ on, if $r+\ell-ik\in L(T)$, for some $i\in\{0,\ldots,m-1\}$ with $ik<\ell$.

To this end, let $\tau^j$ denote the Romeo strategy that plays $1+n+k$ during its first $j$ moves and $1+n$ afterwards. 

Let thus $w=r+\ell-ik$. 
Since $w\in L(T)$, $\sigma$ wins on $w$. In particular, $\sigma$ wins against strategy $\tau^i$. It is easy to see that the play $\Pi(\sigma,\tau^i,w)$ reaches a configuration with the same profile as $K$ and thus, as explained before, $K$. 
Since $\sigma$ wins on $w$, it thus wins all plays from $K$ on. 

Next, we are going to construct from $\sigma$ a strategy $\sigma'$ that wins on $L$. To this end, let $w\in L(T)$ be such that $\Pi(\sigma,\tau_0,w)$ has at least $m$ Call-moves. Let $(m,r,\ell)$ be the profile of the configuration $K$ that is reached in $\Pi(\sigma,\tau_0,w)$ after the $m$-th Call move.  By the above claim, $\sigma$ wins from every configuration $(m,r,\ell')$, for which $r+\ell'-ik\in L(T)$, for some $i\in\{0,\ldots,m-1\}$ with $ik<\ell'$. The $T$-state of $K$ is $r'=r \mod m$.

Strategy $\sigma'$ is defined as follows. 
It first does $r'$ Read moves, yielding some configuration $K'$ with $T$-state $r'$.
 From there on, it copies the strategy of $\sigma$ in plays that start from $K$. More precisely, let $\alpha$ be the history string of $K$ and $\alpha'=r'$ be the history string of $K'$. Then, for every history string $\beta$, $\sigma'$ plays on $\alpha'\cdot\beta$ like $\sigma$ on  $\alpha\cdot\beta$. We can conclude that $\sigma'$ not only wins on every string $v\in L(T)$ but also on every string of the form $v+ik$ with $v\in L(T)$, $v\ge r'$ and $i\in\{0,\ldots,m-1\}$. 

Therefore, $\sigma'$ is a winning strategy for $L$, and \ref{it:unaryInducedGood} holds.
	
	\ \\
	\noindent \textit{Proof of (\ref{it:unaryNoneBetter}):}

\noindent We show that there is no strategy for Juliet that wins on $L$ and some word $w\not\in L$ of length at least $(k+1)m$.

Let $d$ be the ultimate number of words in $L$ per $k$ words, that is let $d:=|L\cap\{km+1,\ldots,km+k\}|$ be the number of words in $L$ per every $k$ adjacent words. Suppose some strategy $\sigma'$ wins on $L$ and on some word $w\notin L$ of length at least $km$, and let $L'=L\cup\{w\}$. In particular, the play $\Pi(\sigma',\tau,v)$ terminates, for every $v\in L'$ and every strategy of Romeo. Towards a contradiction, we prove by induction that, for every $t\ge 1$, $\Pi(\sigma',\tau_1,w)$ reaches\footnote{We recall that $\tau_1$ is the strategy of Romeo that always chooses the answer string $1+n+k$.}  some configuration $K_m$ with profile $(t,r,\ell)$, for some $r,\ell$. Clearly, this contradicts the assumption that $\Pi(\sigma',\tau_1,w)$ terminates.

For $m=1$, the claim follows since $w\not\in L(T)$.  

For the inductive step, let $K$ be the first configuration  in $\Pi(\sigma',\tau_1,w)$ with profile $(m,r,\ell)$, for some $r,\ell$. Clearly, the move that leads to $K$ is a $CR^n$-move and $\ell> k$, since Romeo played $1+n+k$. If $\Pi(\sigma',\tau_1,w)$ makes a further $CR^n$-move after $K$, we are done with the inductive step, thus we assume it does not. Since $\ell> k$, the play must therefore reach a configuration $K''$ with profile $(m,r',k)$, for some $r'$. Let  $\alpha$ be its history string.

By $w_1,\ldots,w_d$ we denote the minimal $d$ words that are longer than $w$, and for which $\sigma'$ is a winning strategy. For every $i\le d$, let $c_i=|w_i|-|w|$. By the choice of $d$ and since $L$ is $k$-periodic beyond $km$, it follows that $c_d\le k-1$. 

Let $\tau'$ denote the strategy that deviates from $\tau_1$ (only) by answering the $CR^n$-move at configuration $K'$ by $1+n$. 
 It is easy to see that, for each $i\le d$, the play $\Pi(\sigma',\tau',w_i)$ reaches, following the same history string $\alpha$,  a configuration $K''_i$ with profile $(t,r',c_i)$. Since we assume that $\sigma'$ wins from all these configurations, and there can be at most $d$ accepting states within the next $k$ states from the state of $T$ for $K''$, $\sigma'$ has to play $CR^n$ at least once during the next $k$ steps. The play  $\Pi(\sigma',\tau_1,w)$ thus reaches a configuration with profile $(t+1,r'',\ell')$ with some $r''$ and $\ell'$, completing the inductive step.
\end{proofof}
} % \full
\subsection{Negative results}

In the previous subsection we showed the strongly positive result that, on non-recursive games or games with finite target languages, the existence of a dominant strategy always implies that there exists a \emph{strongly regular} dominant strategy. We now turn to the negative results of this section and show that the above implication does not generalise to arbitrary games or to undominated strategies; in fact, the following theorem states that even much weaker statements do not hold in general.

\begin{theorem}\label{theo:wrong}
  \begin{enumerate}[(a)]%[label=(\alph*)]
  \item 	There exists a game $G_1$ with a regular dominant strategy but no forgetful dominant strategy.\label{it:negA}
  \item There exists a game $G_2$ with a forgetful regular dominant strategy but no strongly regular dominant strategy.\label{it:negB}
    \item The statements (a) and (b) hold also when ``dominant'' is replaced by ``undominated''. In this case, $G_1$ and $G_2$ can even be chosen as non-recursive games with a finite target language.
  \end{enumerate}
\end{theorem}

\full{
\begin{proof}%[of Theorem~\ref{theo:wrong}]
  \begin{enumerate}[(a)]%[label=(\alph*)]
  \item The game $G_1=(\Sigma,R,T)$ with $\Sigma=\{a\}$, the only replacement rule being $a\to aa$ and $L(T)=\{a^k\mid k\ge 2\}$ has the stated property. The strategy plays Call exactly if it has not seen any symbol $\call{a}$. Since this strategy wins on every word, it is dominant. However, a forgetful strategy that plays Call on the first symbol $a$ is bound to play Call forever, and therefore does not win on any word. On the other hand, a forgetful strategy that plays Read on the first symbol $a$ does not win on the word $a$ and is therefore not dominant, either.
  \item Let $G_2=(\Sigma,R,T)$ with $\Sigma=\{a,b,c,d\}$, rule set $R$ given by 
\begin{align*}
&a\to b\\
&c\to ac\\
&d\to bad
\end{align*}
and the target language automaton $T=(Q,\Sigma,\delta_T,q_0,F)$ depicted in Figure~\ref{fig:dominant45wrongT}. We claim that the regular forgetful strategy~$\sigma_{\mathcal A}$ based on the automaton $\mathcal A$ shown in Figure~\ref{fig:dominant45wrongA} fulfils $W(\sigma_{\mathcal A} )=\Sigma^*$ and is thus dominant.
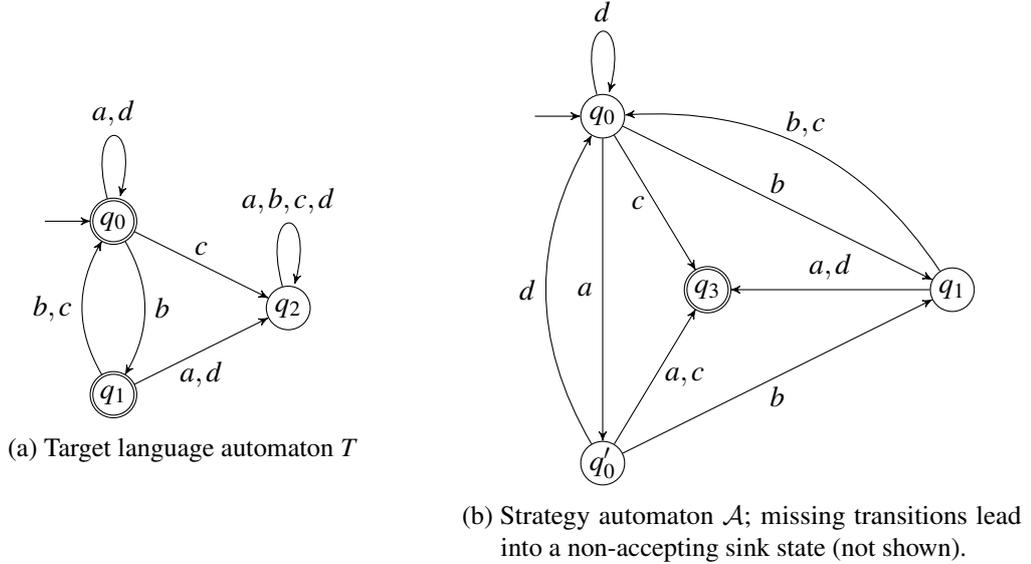
\begin{figure}[ht]
\centering
\begin{subfigure}{.5\textwidth}
\centering
\begin{tikzpicture}[xscale=2.3,yscale=2.3,>=stealth',initial text=,initial distance=1.5ex,every state/.style={inner sep=0pt,minimum size=1.5em}]
\node[state,initial,accepting] at (0,1) (q0) {$q_0$};
\node[state,accepting] at (0,0) (q1) {$q_1$};
\node[state] at (1,0.5) (q2) {$q_2$};
\path[->] (q0)	edge[loop above]	node[above] {$a,d$} ()
				edge[bend left]		node[right] {$b$} (q1)
				edge				node[above] {$c$} (q2)
		  (q1)	edge[bend left]		node[left]	{$b,c$} (q0)
				edge				node[below] {$a,d$} (q2)
		  (q2)	edge[loop above]	node[above] {$a,b,c,d$} ();
\end{tikzpicture}
\caption{Target language automaton $T$}\label{fig:dominant45wrongT}
\end{subfigure}% Do not remove this comment as otherwise it will cause a spurious blank space to be added, the total length will surpass \textwidth and the figures will end up not side-by-side
\begin{subfigure}{.5\textwidth}
\centering
\begin{tikzpicture}[xscale=2.3,yscale=2.3,>=stealth',initial text=,initial distance=1.5ex,every state/.style={inner sep=0pt,minimum size=1.5em}]
\node[state,initial] at (0,2) (q0) {$q_0$};
\node[state] at (0,0) (q0') {$q'_0$};
\node[state] at (2,1) (q1) {$q_1$};
\node[state,accepting] at (0.6,1) (q3) {$q_3$};
\path[->] (q0)	edge[loop above]	node[above] {$d$} ()
				edge		node[left] {$a$} (q0')
				edge				node[above] {$b$} (q1)
				edge	node[left] {$c$} (q3)
		  (q0')	edge[bend left]		node[left]	{$d$} (q0)
				edge				node[below] {$b$} (q1)
				edge				node[right] {$a,c$} (q3)
		  (q1)	edge[bend right]		node[above] {$b,c$} (q0)
				edge				node[above] {$a,d$} (q3);

\end{tikzpicture}
\caption{Strategy automaton $\mathcal A$; missing transitions lead into a non-accepting sink state (not shown).} \label{fig:dominant45wrongA}
\end{subfigure}
\caption{Automata used in the proof of Theorem~\ref{theo:wrong}~(b)}
\end{figure}
Indeed, by induction on the length of $w$, the following is easy to show: for each input word $w$, the strategy $\sigma_\calA$ yields a terminating play with a final string $u$ such that: 
$\delta_T(q_0,u)=\delta_\calA(q_0,u)$, if $u$ does not end with $a$, and $\delta_T(q_0,u)=q_0$ and $\delta_\calA(q_0,u)=q_0'$, otherwise.

However, for a strategy automaton $\mathcal B=(Q\cup\{\Call\},\Sigma,\delta_{\mathcal B},q_0,\{\Call\})$ of a strongly regular strategy $\sigma_{\mathcal B}$, it holds that $W(\sigma_{\mathcal B})\subsetneq\Sigma^*$, and thus no such $\sigma_{\mathcal B}$ is dominant.
For a proof of this claim, we can assume that $\delta_{\mathcal B}(q_0,c)=\delta_{\mathcal B}(q_1,a)=\delta_{\mathcal B}(q_1,d)=\Call$ since otherwise $\sigma_{\mathcal B}$ would lose on $c$, $ba$ or $bd$. If $\delta_{\mathcal B}(q_0,a)=q_0$, then the play of $\sigma_{\mathcal B}$ on $ac$ is infinite. On the other hand, if $\delta_{\mathcal B}(q_0,a)=\Call$, then play of $\sigma_{\mathcal B}$ on $ad$ is infinite.
\item For the analogon of (a), let $G_1=(\Sigma,R,T)$ over the alphabet $\Sigma=\{a,b,c,d,e\}$ with the set of rules $R$ given by
	\begin{align*}
		&a\to b+c\\
		&b\to cd\\
		&c\to e
	\end{align*}
	and target language $L(T)=\{e,cd\}$. It is not hard to see that the strategy that plays Call on $(\epsilon,a)$, $(\epsilon,b)$, $(\epsilon,c)$, $(\call{a},b)$ and $(\call{a},c)$ and Read otherwise is regular, wins on $\{a,b,c,e\}$, and is undominated. 

Towards a contradiction, let us assume that $\sigma$ is a forgetful strategy for $G_1$ that is undominated. We show first that $\sigma$ does not win on $a$. Suppose $\sigma$ does win on $a$. Then $\sigma(\epsilon,a)=\sigma(\epsilon,b)=\Call$. If $\sigma(\epsilon,c)=\Read$, then the play
$(\epsilon,a),(\call a,c),(\call ac,\epsilon)$
of $\sigma$ on $a$ is losing -- contradiction. If $\sigma(\epsilon,c)=\Call$, then the play
%\begin{align*}
$$(\epsilon,a),(\call a,b),(\call a\call b,cd),(\call a\call b\call c,ed),(\call a\call b\call ce,d),(\call a\call b\call ced,\epsilon)$$
%\end{align*}
of $\sigma$ on $a$ is losing -- contradiction.

It is also easy to see that no strategy can win on a word that starts with $a$ and has length at least $2$. A one-pass strategy $\sigma'$ that wins on all words from $W(\sigma)$ and also on $a$ can be defined by
\[\sigma'(\alpha, f)=
\begin{cases}
\Call,&\text{if $\alpha f\in\{a, \call a b, \call a c\}$,}\\
\Read,&\text{if $\alpha f=\call a\call b c$,}\\
\sigma(\alpha, f),&\text{otherwise}
\end{cases}\]
for $\alpha\in\hs$, $f\in\Sigma_f=\{a,b,c\}$. Since $W(\sigma)\subsetneq W(\sigma)\dotCup\{a\}=W(\sigma')$, it follows that $\sigma$ is not undominated.

For the analogon of (b), consider the game $G_2=(\Sigma,R,T)$ over the alphabet $\Sigma=\{a,b,c\}$ with rules given by
\begin{align*}
	&a\to bb+cbc\\
	&b\to cc
\end{align*}
and target language automaton $T$ as per Figure~\ref{fig:optimal45wrongT}. The finite target language is $L(T)=\{bbc,bcc,cbc,ccc\}$. 

\begin{figure}[ht]
	\centering
	\subcaptionbox{Target language automaton $T$\label{fig:optimal45wrongT}}[.5\textwidth]{
		\centering
		\begin{tikzpicture}[xscale=2.3,yscale=2.3,>=stealth',initial text=,initial distance=1.5ex,every state/.style={inner sep=0pt,minimum size=1.5em}]
		\node[state,initial] at (0,0) (q0) {$q_0$};
		\node[state] at (0.8,0) (q1) {$q_1$};
		\node[state] at (1.6,0) (q2) {$q_2$};
		\node[state,accepting] at (2.4,0) (q3) {$q_3$};
		\node[state] at (1.2,1) (q4) {$q_4$};
		\path[->] (q0)	edge node[above] {$b,c$} (q1)
		edge[bend left] node[left] {$a$} (q4)
		(q1)	edge node[above] {$b,c$} (q2)
		edge node[left] {$a$} (q4)
		(q2)	edge node[above] {$c$} (q3)
		edge node[right] {$a,b$} (q4)
		(q3)	edge[bend right] node[right] {$a,b,c$} (q4)
		(q4)	edge[out=135,in=45,loop] node[above] {$a,b,c$} ();
		\end{tikzpicture}}
			\subcaptionbox{Strategy automaton $\mathcal A$\label{fig:optimal45wrong4}}[.5\textwidth]{
		\centering
		\begin{tikzpicture}[xscale=2.3,yscale=2.3,>=stealth',initial text=,initial distance=1.5ex,every state/.style={inner sep=0pt,minimum size=1.5em}]
%		\clip (-0.44,-0.13) rectangle (2.63,1.7);
		\node[state,initial] at (0,0) (p0) {$q_0$};
		\node[state] at (0,2) (p1) {$q_1$};
		\node[state] at (3,0) (p2) {$q_2$};
		\node[state] at (3,2) (p3) {$q_3$};
		\node[state] at (2.2,1) (p4) {$q_4$};
		\node[state,accepting] at (1,1) (p5) {$q_5$};
		\path[->] (p0)	edge node[left] {$b$} (p1)
		edge node[below] {$c$} (p2)
		edge node[below] {$a$} (p5)
		(p1)	edge node[above] {$c$} (p3)
		edge node[right] {$a,b$} (p5)
		(p2)	edge node[below right] {$b,c$} (p3)
		edge node[below] {$a$} (p5)
		(p3)	edge node[above] {$c$} (p4)
		edge node[above] {$a,b$} (p5)
		(p4)	edge[out=-45,in=35,loop] node[right] {$c$} ()
		edge node[below right] {$a,b$} (p5);
		\end{tikzpicture}}
	\caption{Automata used in the proof of Theorem~\ref{theo:wrong}~(c)}
\end{figure}

We claim that the regular forgetful strategy $\sigma_{\mathcal A}$ based on the automaton $\mathcal A$ shown in Figure~\ref{fig:optimal45wrong4} is undominated. It can be checked easily that $W(\sigma_{\mathcal A})=\{a,bb,bcc,cbc,ccc\}$. The only words that are not in $W(\sigma_{\mathcal A})$ even though a strategy exists that wins on them are $bc$, $bbc$ and $cb$. However, a one-pass strategy that wins on $bc$ cannot win on $bcc$; a strategy that wins on $bbc$ cannot win on $bb$; a strategy that wins on $cb$ cannot win on $cbc$. Therefore, every one-pass strategy that wins on a word on which $\sigma_{\mathcal A}$ does not win also loses on a word on which $\sigma_{\mathcal A}$ wins. Hence, there exists no strategy $\sigma$ with $W(\sigma_{\mathcal A})\subsetneq W(\sigma)$, which means that $\sigma$ is undominated.

It remains to show that there exists no undominated strongly regular strategy for $G$. 
Let $B=(Q\cup\{\Call\},\Sigma,\delta_{\mathcal B},q_0,\{\Call\})$ be the strategy automaton of a strongly regular strategy $\sigma_{\mathcal B}$. Observe that $\sigma_{\mathcal B}$ cannot win on $a$: In order to win on $a$, a strategy must play Call in the configurations $(\epsilon,a,\epsilon)$ and $(\call ab,b, \epsilon)$ and it must play Read in the configurations $(\call a,b,b)$ and $(\call ac,b,c)$. But a strongly regular strategy must play the same in the configurations $(\call ab,b,\epsilon)$ and $(\call ac,b,c)$, because both have $T$-state $q_1$ and current symbol $b$. However, we can find a one-pass strategy $\sigma$ with $W(\sigma_{\mathcal B})\subsetneq W(\sigma_{\mathcal B})\dotCup\{a\}=W(\sigma)$. Indeed, such a strategy can be defined by
\begin{align*}
\sigma(\alpha,f)=
\begin{cases}
\Call,&\text{if $(\alpha,f)=(\epsilon,a)$ or $(\alpha,f)=(\call ab,b)$,}\\
\Read,&\text{if $(\alpha,f)=(\call a,b)$ or $(\alpha,f)=(\call ac,b)$,}\\
\sigma_{\mathcal B}(\alpha,f),&\text{otherwise}
\end{cases}
\end{align*}
for $\alpha\in\hs$ and $f\in\Sigma_f=\{a,b\}$. Hence, $\sigma_{\mathcal B}$ is not undominated.
\end{enumerate}
\vspace{-\baselineskip}
\end{proof}
 } % \full

\section{Complexity results for regular strategies}\label{cha:complexity}
For classes of games where (weakly) dominant or undominated strategies are not guaranteed to exist, it would be desirable to have algorithms determining, given a game $G$, whether a strategy of a certain type exists on $G$, or whether a given strategy is (weakly) dominant or undominated on $G$. Unfortunately, in general, these problems are not even known to be decidable.

In this section, we start the investigation of algorithmic problems related to one-pass strategies by studying the complexity of three somewhat more restricted types of problems. The first is to test for a given game, input word and strategy whether the strategy wins on the input word. The second is similar, but the strategy is not part of the input: The question in this case is whether a strategy exists that wins on a given input word in a given game. The third problem is relevant in the context of determining dominant or undominated strategies: Given a game and two strategies, is one strategy better than the other? That is, we examine the following decision problems for automata-based classes $\mathcal{S}$ of strategies: 

%FOR CORRECT DISPLAY IN scrartcl:
%Uncomment this block (and comment out the one after)

\newcommand{\myboxi}[3]{\node[action,fill=blue!10] at (#1) {\begin{minipage}{#2}%
#3\end{minipage}};}

\renewcommand{\algproblem}[5][0.3\textwidth]
{%
\myboxi{#2}{#1}{\small
\centerline{\textbf{#3}}
    \begin{description}
    \item[Input:]  #4
    \item[Question:] #5
    \end{description}
}}

\centerline{
\scalebox{0.9}{
\begin{tikzpicture}
\algproblem{0,1}{IsWinning}{Game $G=(\Sigma,R,T)$, word $w\in\Sigma^*$, strategy automaton $\mathcal A$ for $G$}{Is $w\in W(\sigma_{\mathcal A})$?}
\algproblem{5,0.8}{ExistsWinning($\mathcal S$)}{Game $G=(\Sigma,R,T)$, \\ word $w\in\Sigma^*$}{Does a strategy $\sigma\in \mathcal S$ for $G$ exist such that $w\in W(\sigma)$?}
\algproblem{10,1}{IsDominated}{Game $G$, strategy automata $\mathcal A_1$ and $\mathcal A_2$ for $G$}{\mbox{ }\newline Is $W(\sigma_{\mathcal A_1})\subseteq W(\sigma_{\mathcal A_2})$?}
\end{tikzpicture}
}
}

%%FOR CORRECT DISPLAY IN LIPICS:
%%Uncomment this block (and comment out the previous)
%
%\newcommand{\myboxi}[3]{\node[action,fill=blue!10] at (#1) {\begin{minipage}{#2}%
%#3\end{minipage}};}
%
%\renewcommand{\algproblem}[5][0.33\textwidth]
%{%
%%\setlist[description]{leftmargin=.5cm}%
%\myboxi{#2}{#1}{\small
%\centerline{\textbf{#3}}
%    \begin{description}
%    \item[Input:]  #4
%    \item[Question:] #5
%    \end{description}
%%\vspace{-5mm\\}
%}}
%
%\centerline{
%\scalebox{0.9}{
%\begin{tikzpicture}
%\algproblem{0,1}{IsWinning}{Game $G=(\Sigma,R,T)$, \\ word $w\in\Sigma^*$, \\ strategy automaton $\mathcal A$ for $G$}{Is $w\in W(\sigma_{\mathcal A})$?} 
%%\hfill%
%\algproblem{5,1}{ExistsWinning($\mathcal S$)}{Game $G=(\Sigma,R,T)$, \\ word $w\in\Sigma^*$}{Does a strategy $\sigma\in \mathcal S$ for $G$ exist such that $w\in W(\sigma)$?}
%%\end{tikzpicture}
%%\hfill%
%%
%%\begin{tikzpicture}
%\algproblem{10,1}{IsDominated}{Game $G$, strategy automata $\mathcal A_1$ and $\mathcal A_2$ for $G$}{\mbox{ }\newline Is $W(\sigma_{\mathcal A_1})\subseteq W(\sigma_{\mathcal A_2})$?}
%\end{tikzpicture}
%}
%}

For the \textsc{IsWinning} and \textsc{IsDominated} problems (and the more general problem whether a given strategy is dominant or undominated), we restrict our attention to regular strategies, since non-regular strategies do not in general have a finite encoding (as there are uncountably many of them). The \textsc{ExistsWinning} problem may be examined for several types of strategies (regular, forgetful, etc.); however, we mostly focus on the set \SReg of strongly regular strategies.% as the most relevant type.

The main results of this section are summarised in the following theorem.
\begin{theorem}\label{thm:complexity}
  \begin{enumerate}[(a)]%[label=(\alph*)]
  \item $\textsc{IsWinning} \in \P$. \label{thm:Pcomplete}
  \item \textsc{ExistsWinning($\tFive$)} is $\NP$-complete. The lower bound already holds for non-recursive games whose replacement languages are finite and prefix-free. \label{thm:NPcomplete}
  \item \textsc{IsDominated} is $\PSPACE$-complete. The lower bound already holds for strongly regular strategies and for non-recursive games whose plays all have depth at most $2$, and whose replacement words all have length $1$.
  \end{enumerate}
\end{theorem}

\full{
The remainder of this section is structured as follows. As an important tool for the upper bounds, we first show that relations describing Romeo's power against a given regular strategy can be computed in polynomial time (Proposition~\ref{prop:strategyAutomatonTests}).
Then we show that, from these relations, an NFA for those words for which Romeo has a winning strategy or can force infinite play can be computed, again in polynomial time (Propositon~\ref{prop:winningSetRegular}). Afterwards the upper bounds of Theorem~\ref{thm:complexity} are shown as Proposition~\ref{prop:complexity-upper}. The lower bounds are then shown as Propositions~\ref{prop:existswinning-lower} and \ref{prop:compare-lower}. The section is wrapped up by some additional results from \cite{Coester15}.

 In this section, we use some additional notation. Let $\mathcal{A} = (Q,\hist{\Sigma},\delta,s,F)$ be  a strategy automaton. For a configuration $(\alpha,u)$ arising in a play following strategy $\sigma_\calA$, we refer to $\delta^*(s,\alpha)$ as its $\calA$-state. 
For each $q \in Q$, we denote by $\mathcal{A}_q = (Q,\hist\Sigma,\delta,q,F)$ the automaton derived from $\mathcal{A}$ by setting its initial state to $q$. 
For a word $w$ and strategies $\sigma$ and $\tau$ of Juliet and Romeo, we write $h(\Pi(\sigma,\tau,w))$ for the final history string of $\Pi(\sigma,\tau,w)$. Whenever we do so we imply that $\Pi(\sigma,\tau,w)$ is finite. For a Romeo strategy $\tau$ and a history string $\alpha$, we denote by $\tau^\alpha$ the strategy given by $\tau^\alpha(\beta,a)=\tau(\alpha\beta,a)$.

\begin{proposition}\label{prop:strategyAutomatonTests}
For a strategy automaton $\mathcal A$ and a game $G=(\Sigma,R,T)$, the following two relations can be computed in polynomial time.
\begin{enumerate}[(a)]%[label=(\alph*)]
\item The set of triples $(q,a,q')$ for which Romeo has a strategy 
$\tau$ such that $\delta^*(q,\alpha) = q'$ holds for 
    $\alpha=h(\Pi(\sigma_{\mathcal A_q},\tau,a))$.
\item The set of pairs $(p,a)$ for which Romeo has a strategy $\tau$ such that $\Pi(\sigma_{\calA_p},\tau,a)$ is infinite.
\end{enumerate}
\end{proposition}

\begin{proof}
Let $\mathcal{A} = (Q,\hist{\Sigma},\delta,s,F)$. 

Towards (a), we inductively define a relation $\Move$ and show that it captures precisely the possible kinds of state-to-state behaviour of sub-plays in which Juliet follows $\sigma_{\mathcal A}$.

Let $S$ be the set of all subexpressions of the regular expressions $R_a$ for all $a\in\Sigma_f$. We will define a relation $\Move \subseteq Q\times (S\cup\Sigma) \times Q$, where a triple $(q,a,q')\in\Move$ with $a\in\Sigma$ shall denote that a play of $\sigma_{\mathcal A_q}$ on $a$ may terminate in state $q'$ (for some strategy of Romeo). We note that the union $S\cup\Sigma$ is not disjoint since replacement words are non-empty by definition, and therefore some alphabet symbols occur as atomic subexpressions.

For $(q,r,q')\in Q\times (S\cup\Sigma) \times Q$, we define inductively that $(q,r,q')\in\Move$ if any of the following conditions holds for some $a\in\Sigma$,  $r_1,r_2\in S$ and $q''\in Q$:

\begin{enumerate}[(a)]
\item $r=a$ and $q'=\delta(q,a) \notin F$ \label{it:winningSetRegularMoveRead}
\item $r=a$, $\delta(q,a) \in F$ and $(\delta(q,\call{a}),R_a,q')\in\Move$\label{it:winningSetRegularMoveCall}
\item $r=\epsilon$ and $q=q'$\label{it:winningSetRegularMoveEps}
\item $r=r_1r_2$ and $(q,r_1,q''), (q'',r_2,q')\in\Move$\label{it:winningSetRegularMoveConc}
\item $r=r_1+r_2$ and $(q,r_1,q')\in\Move$
\item $r=r_1+r_2$ and $(q,r_2,q')\in\Move$
\item $r=r_1^*$ and $q'=q$\label{it:winningSetRegularMoveKleeneEps}
\item $r=r_1^*$ and $(q,r_1,q''), (q'',r_1^*,q')\in\Move$\label{it:winningSetRegularMoveKleeneMain}
\end{enumerate}
The semantics of the $\Move$ relation is characterised by the following claim.

\begin{myclaim}
  \label{cl:winningSetRegularMove}
  For $q,q'\in Q$ and $r\in S\cup\Sigma$, the following are
  equivalent:
  \begin{itemize}
  \item $(q,r,q')\in\Move$
  \item There exist $w\in L(r)$ and a Romeo strategy $\tau$ such that  $\delta^*(q,\alpha) = q'$ holds for 
    $\alpha=h(\Pi(\sigma_{\mathcal A_q},\tau,w))$.
  \end{itemize}
\end{myclaim}
We show both implications separately by induction.
\begin{description}
\item[``$\Longrightarrow$'':] The proof of the implication from the first to the second statement is by structural induction on elements of $\Move$.

If $r=a\in\Sigma$ and $q'=\delta(q,a) \notin F$, then the second statement holds for $w=a$ and an arbitrary Romeo strategy $\tau$, since $\sigma_{\mathcal A_q}$ plays $\Read$ on $a$.

If $r=a\in\Sigma$, $\delta(q,a)\in F$ and $(\delta(q,\call{a}),R_a,q')\in\Move$, then by the induction hypothesis there exist $w'\in L(R_a)$ and a Romeo strategy $\tau'$ such that $h(\Pi(\sigma_{\mathcal A_{\delta(q,\call{a})}},\tau',w'))=\beta$ with $\delta^*(\delta(q,\call{a})),\beta) = q'$.
The second statement holds for $w=a$ and any Romeo strategy $\tau$ with $\tau(\epsilon,a)=w'$ and $\tau(\call a\alpha,b)=\tau'(\alpha,b)$ for $\alpha\in\hist{\Sigma}^*$ and $b\in\Sigma_f$. That is, if the input word is $a$ and Juliet's first move is $\Call$, then $\tau$ replaces $a$ by $w'$ and continues to play like $\tau'$ plays on the input word $w'$.

If $r=\epsilon$ and $q=q'$, then the second statement holds for $w=\epsilon$ and any $\tau$.

If $r=r_1r_2$ and $(q,r_1,q''), (q'',r_2,q')\in\Move$, then by the induction hypothesis there exist words $w_1\in L(r_1)$, $w_2\in L(r_2)$ and Romeo strategies $\tau_1$ and $\tau_2$ for which it holds that  \mbox{$h(\Pi(\sigma_{\mathcal A_q},\tau_1,w_1))=\beta$} with $\delta^*(q,\beta) = q''$
and $h(\Pi(\sigma_{\mathcal A_{q''}},\tau_2,w_2))=\beta'$ with $\delta^*(q'',\beta') = q'$.
The second statement holds for $w=w_1w_2$ and the Romeo strategy that plays like $\tau_1$ on $w_1$ and like $\tau_2$ thereafter.
	
The remaining cases are similar.

\item[``$\Longleftarrow$'':] Let $w\in L(r)$ and a Romeo strategy $\tau$ be such that $\delta^*(q,\alpha) = q'$ for $\alpha=h(\Pi(\sigma_{\mathcal A_q},\tau,w))$. 
Let $x = \natural \alpha \in\Sigma^*$ be the final string of this play. For $a\in\Sigma_f$, let $n_a$ be the number of Calls of $a$ occurring during this play. We prove the implication by induction on $\abs{x}+\abs{r}+\sum_{a\in\Sigma_f}n_a\abs{R_a}$. We distinguish several cases and deduce the first statement in each case either directly or by applying the induction hypothesis.

\begin{itemize}
\item If $r=\epsilon$, then $w=\epsilon$ and therefore $q'=q$. Hence, $(q,r,q')\in\Move$ follows directly from the definition, since condition~\eqref{it:winningSetRegularMoveEps} is satisfied.
\item If $r=a\in\Sigma$ and $\delta(q,a)\notin F$, then $w=a$ and $q'=\delta(q,a)$. Again $(q,r,q')\in\Move$ follows directly, since condition~\eqref{it:winningSetRegularMoveRead} of the definition is satisfied.
\item If $r=a\in\Sigma$ and $\delta(q,a)\in F$, then $w=a$ and for $\alpha'=h(\Pi(\sigma_{\mathcal A_{\delta(q,\call{a})}},\tau^{\call a},\tau(\epsilon,a)))$ it holds that $\delta^*(\delta(q,\call{a}),\alpha') = q'$.
We can apply the induction hypothesis and get $({\delta(q,\call{a})},R_a,q')\in\Move$. Thus, $(q,r,q')\in\Move$ holds by condition~\eqref{it:winningSetRegularMoveCall} of the definition.
\item If $r=r_1^*$ then we can write $w=w_1w_2\dots w_n$ with $n\in\mathbb N_0$ and $w_j\in L(r_1)\setminus\{\epsilon\}$. If $n=0$, then $w=\epsilon$ and $q=q'$, so $(q,r,q')\in\Move$ by condition~\eqref{it:winningSetRegularMoveKleeneEps}. If $n\ge1$, let $\alpha_1\df h(\Pi(\sigma_{\mathcal A_q},\tau,w_1))$ and let $q''=\delta(q,\alpha_1)$.
Since $\abs{r_1^*}=\abs{r_1}+1>\abs{r_1}$, we can apply the induction hypothesis, which yields $(q,r_1,q'')\in\Move$. Observe that $h(\Pi(\sigma_{\mathcal A_{q''}},\tau^{\alpha_1},w_2w_3\dots w_n))=\alpha_2$ with $\delta^*(q'',\alpha_2) = q'$,
and its final string is the string $x'$ with $\natural\alpha_1 x'=x$. From $\abs{\natural\alpha_1}\ge \abs{w_1}>0$ it follows that $\abs{x}>\abs{x'}$, so another application of the induction hypothesis yields $(q'',r_1^*,q')\in\Move$. We get $(q,r,q')\in\Move$ from condition~\eqref{it:winningSetRegularMoveKleeneMain} of the definition.
\item The cases $r=r_1r_2$ and $r=r_1+r_2$ are similar to and slightly easier than the previous case.
\end{itemize}
\end{description}
This concludes the proof of Claim~\ref{cl:winningSetRegularMove}.

The basic idea for the detection of infinite plays is as follows: for a play to be infinite there must be a function symbol $a\in\Sigma_f$ and a state $q\in Q$ such that $a$ is called in $q$ and in the ``sub-play'' on the replacement word of $a$ there occurs another configuration in which $\calA$ has state $q$ and the current symbol is $a$. To detect the possibility of such a play, we introduce a relation $\Next\subseteq Q\times(S\cup\Sigma)\times Q\times\Sigma$. A tuple $(q,r,q',a)$ shall be in $\Next$ if and only if there exists a play of $\sigma_{\mathcal A_q}$ on some word from $L(r)$ that contains a configuration with $\mathcal A_q$-state $q'$ and current symbol $a$. 

To this end, we define inductively, for $(q,r,q',a)\in Q\times(S\cup\Sigma)\times Q\times\Sigma$, that $(q,r,q',a)\in\Next$ if any of the following conditions holds for some $b\in\Sigma$, $r_1,r_2\in S$ and $q''\in Q$:

\begin{enumerate}[(a)]%[resume]
\item $r=a$ and $q=q'$ \label{it:winningSetRegularNextInitial}
\item $r=b$, $\delta(q,b) \in F$ and $(\delta(q,\call{b}),R_b,q',a)\in\Next$ \label{it:winningSetRegularNextCall}
\item $r=r_1r_2$ and $(q,r_1,q',a)\in\Next$ 
\item $r=r_1r_2$, $(q,r_1,q'')\in\Move$ and $(q'',r_2,q',a)\in\Next$
\item $r=r_1+r_2$ and $(q,r_1,q',a)\in\Next$
\item $r=r_1+r_2$ and $(q,r_2,q',a)\in\Next$
\item $r=r_1^*$, $(q,r_1^*,q'')\in\Move$ and $(q'',r_1,q',a)\in\Next$ \label{it:winningSetRegularNextKleene}
\end{enumerate}

\begin{myclaim}
  \label{cl:winningSetRegularNext}
  For $q,q'\in Q$, $r\in S\cup\Sigma$ and $a\in\Sigma$, the following
  are equivalent:
  \begin{itemize}
  \item $(q,r,q',a)\in\Next$
  \item There exist $w\in L(r)$ and a Romeo strategy $\tau$ such that
    $\Pi(\sigma_{\mathcal A_q},\tau,w)$ contains a configuration with
    $\mathcal A_q$-state $q'$ and current symbol $a$.
  \end{itemize}
\end{myclaim}

\begin{description}
\item[``$\Longrightarrow$'':] The proof is by structural induction on elements of $\Next$ and uses the same techniques that were already in the proof of the implication ``$\Longrightarrow$'' of Claim~\ref{cl:winningSetRegularMove}. 
\item[``$\Longleftarrow$'':] Let $w\in L(r)$ and $\tau$ be a Romeo strategy such that $\Pi(\sigma_{\mathcal A_q},\tau,w)$ contains a configuration with $\mathcal A_q$-state $q'$ and current symbol $a$ and $w$ is of minimal length with these properties. Let $\alpha$ be the history string of the first such configuration. For $b\in\Sigma_f$ let $n_b$ be the number of occurrences of $\call b$ in $\alpha$, \ie the number of Calls of $b$ before this configuration is reached. The proof is by induction on $\abs{r}+\sum_{b\in\Sigma_f}n_b\abs{R_b}$. Note that $w\ne\epsilon$, because the only configuration in a play on $\epsilon$ is $(\epsilon,\epsilon)$, which does not have current symbol $a$. So in particular $r\ne\epsilon$. There are several cases.
\begin{itemize}
\item If $r=a$ and $q=q'$, then $(q,r,q',a)\in\Next$ by condition \eqref{it:winningSetRegularNextInitial} of the definition.
\item If $r=b\in\Sigma$, but $b\ne a$ or $q\ne q'$, then $w=b$ and the initial configuration $(\epsilon,b)$ does not have both $\mathcal A_q$-state $q'$ and current symbol $a$. It is necessary that $\delta(q,b)\in F$ because otherwise also the second and last configuration $(b,\epsilon)$ does not have current symbol $a$. Thus, $\alpha=\call b\beta$ for some $\beta\in\hist{\Sigma}^*$ and $\Pi(\sigma_{\mathcal A_{\delta(q,\call b)}},\tau^{\call b},\tau(\epsilon,b))$ must have a configuration with $\mathcal A_{\delta(q,\call b)}$-state $q'$ and current symbol $a$. By the induction hypothesis, $(\delta(q,\call b),R_b,q',a)\in\Next$. Hence, $(q,r,q',a)\in\Next$ by condition~\eqref{it:winningSetRegularNextCall} of the definition.
\item If $r=r_1^*$, then, since $w\ne\epsilon$ by the argument above, $w=w_1\dots w_n$ for some $n\ge1$ and $w_i\in L(r_1)\setminus\{\epsilon\}$. By minimality of the length of $w$, the play $\Pi(\sigma_{\mathcal A_q},\tau,w_1\dots w_{n-1})$ must be finite. Let $\beta$ be its history string, and let $q'' = \delta^*(q,\beta)$.
By Claim~\ref{cl:winningSetRegularMove}, $(q,r_1^*,q'')\in\Move$. 
Moreover, $\Pi(\sigma_{\mathcal A_{q''}},\tau^{\beta},w_n)$ contains a configuration with $\mathcal A_{q''}$-state $q'$ and current symbol $a$. By the induction hypothesis, $(q'',r_1,q',a)\in\Next$. It follows from condition~\eqref{it:winningSetRegularNextKleene} of the definition that $(q,r,q',a)\in\Next$.
\item The cases $r=r_1r_2$ and $r=r_1+r_2$ are similar.
\end{itemize}
This concludes the proof of Claim \ref{cl:winningSetRegularNext}.
\end{description}

Finally, we define the relation $\Inf\in Q\times(\Sigma\cup E)$ that shall contain a pair $(q,r)$ if and only if an infinite play of $\sigma_{\mathcal A_q}$ on some word from $L(r)$ exists. For $q\in Q$ and $r\in S\cup\Sigma$, we define inductively that $(q,r)\in\Inf$ if any of the following conditions holds for some $a\in\Sigma$, $r_1,r_2\in S$ and $q\in Q$:

\begin{enumerate}[(a)]%[resume]
\item $r=a$, $\delta(q,a)\in F$ and $(\delta(q,\call{a}),R_a,q,a)\in\Next$ \label{it:winningSetRegularInfBase}
\item $r=a$, $\delta(q,a)\in F$ and $(\delta(q,\call{a}),R_a)\in\Inf$ \label{it:winningSetRegularInfCall}
\item $r=r_1r_2$ and $(q,r_1)\in\Inf$
\item $r=r_1r_2$, $(q,r_1,q')\in\Move$ and $(q',r_2)\in\Inf$
\item $r=r_1+r_2$ and $(q,r_1)\in\Inf$
\item $r=r_1+r_2$ and $(q,r_2)\in\Inf$
\item $r=r_1^*$, $(q,r_1^*,q')\in\Move$ and $(q',r_1)\in\Inf$
\end{enumerate}

\begin{myclaim}
  \label{cl:winningSetRegularInf}
  For $q\in Q$ and $r\in S\cup\Sigma$, the following are equivalent:
  \begin{itemize}
  \item $(q,r)\in\Inf$
  \item There exist $w\in L(r)$ and a Romeo strategy $\tau$ such that
    $\Pi(\sigma_{\mathcal A_q},\tau,w)$ is infinite.
  \end{itemize}
\end{myclaim}
\begin{description}
\item[``$\Longrightarrow$'':] The proof is by structural induction on elements of $\Inf$.

We first consider the case that $r=a\in\Sigma$, $\delta(q,a)\in F$ and $(\delta(q,\call{a}),R_a,q,a)\in\Next$. The idea of a strategy for Romeo that achieves infinite play against $\sigma_{\mathcal A_q}$ on $a$ is to reach a configuration with $\mathcal A_q$-state $q$ and current symbol $a$ (like the initial configuration), and then to repeat the moves from the beginning so that such a configuration occurs over and over again.

Let $w\in L(R_a)$ and a Romeo strategy $\tau$ be such that the play $\Pi(\sigma_{\mathcal A_{\delta(q,\call{a})}},\tau,w)=(K_0,K_1,\dots)$ contains a configuration $K_n$ with $\mathcal A_{\delta(q,\call{a})}$-state $q$ and current symbol $a$. Such $w$ and $\tau$ exist by Claim~\ref{cl:winningSetRegularNext}. For $i=0,\dots,n$, let $\alpha_i\in\hist{\Sigma}^*$ and $v_i\in\Sigma^+$ be such that $K_i=(\alpha_i,v_i)$. Thus, $\alpha_0=\epsilon$, $v_0=w$, $\delta^*(q,\call a\alpha_n)=q$ and $v_n=av$ for some $v\in\Sigma^*$. Let $\tau'$ be any Romeo strategy with $\tau'((\call a\alpha_n)^k,a)=w=v_0$ and $\tau'((\call a\alpha_n)^k\call a\alpha_i,b)=\tau(\alpha_i,b)$ for each $k\in\mathbb N_0$, $b\in\Sigma_f$ and $i=0,\dots,n-1$. We claim that $\Pi(\sigma_{\mathcal A_q},\tau',a)$ is the infinite sequence
\begin{align*}
&((\call a\alpha_n)^0,av^0),((\call a\alpha_n)^0\call a\alpha_0,v_0v^0),((\call a\alpha_n)^0\call a\alpha_1,v_1v^0),\dots,((\call a\alpha_n)^0\call a\alpha_{n-1},v_{n-1}v^0),\\
&((\call a\alpha_n)^1,av^1),((\call a\alpha_n)^1\call a\alpha_0,v_0v^1),((\call a\alpha_n)^1\call a\alpha_1,v_1v^1),\dots,((\call a\alpha_n)^1\call a\alpha_{n-1},v_{n-1}v^1),\\
&((\call a\alpha_n)^2,av^2),\dots.
\end{align*}
The initial configuration is correct, since $((\call a\alpha_n)^0,av^0)=(\epsilon,a)$. By induction on $k\in\mathbb N_0$ it is clear that $\delta^*(q,(\call a\alpha_n)^k)=q$. Therefore,
\begin{align*}
\delta^*(q,(\call a\alpha_n)^ka)=\delta(\delta^*(q,(\call a\alpha_n)^k),a)=\delta(q,a)\in F,
\end{align*}
so $\sigma_{\mathcal A_q}$ plays Call in a configuration $((\call a\alpha_n)^k,av^k)$. The next configuration is therefore
\begin{align*}
((\call a\alpha_n)^k\call a,\tau'((\call a\alpha_n)^k,a)v^k)=((\call a\alpha_n)^k\call a,wv^k)=((\call a\alpha_n)^k\call a\alpha_0,v_0v^k),
\end{align*}
which conforms to the sequence above.

Consider now a configuration $K=((\call a\alpha_n)^k\call a\alpha_i,v_iv^k)$ for $k\in\mathbb N_0$ and $i\in\{0,\dots,n-1\}$. We can write $v_i=b_iv_i'$ where $b_i\in\Sigma$ and $v_i'\in\Sigma^*$. The strategy $\sigma_{\mathcal A_q}$ plays Call if it holds that $\delta^*(q,(\call a\alpha_n)^k\call a\alpha_ib_i)\in F$, \ie if $\delta^*(\delta(q,\call a),\alpha_ib_i)\in F$, which is equivalent to $\sigma_{\mathcal A_{\delta(q,\call a)}}$ playing Call in the configuration $K_i=(\alpha_i,b_iv_i')$. In that case, $\alpha_{i+1}=\alpha_i\call{b_i}$ and $v_{i+1}=\tau(\alpha_i,b_i)v_i'=\tau'((\call a\alpha_n)^k\call a\alpha_i,b_i)v_i'$. The configuration after $K$ is then $((\call a\alpha_n)^k\call a\alpha_{i+1},v_{i+1}v^k)$, as above. Note that this also conforms to the sequence for $i=n-1$ since $v_n=av$. The other case is that $\sigma_{\mathcal A_q}$ plays Read and in $K$ and $\sigma_{\mathcal A_{\delta(q,\call a)}}$ plays Read in $K_i$. This case is analogous. We have shown that if condition \eqref{it:winningSetRegularInfBase} from the definition of $\Inf$ holds, then $\Pi(\sigma_{\mathcal A_q},\tau',a)$ is the infinite sequence from above. This finishes the base case of the structural induction.

The inductive cases are considerably easier. If $r=a\in\Sigma$, $\delta(q,a) \in F$ and $(\delta(q,\call a),R_a)\in\Inf$, then by the induction hypothesis there exist $w\in L(R_a)$ and a Romeo strategy $\tau$ such that $\Pi(\sigma_{\mathcal A_{\delta(q,\call a)}},\tau,w)$ is infinite. Consider a Romeo strategy $\tau'$ with $\tau'(\epsilon,a)=w$ and $\tau'(\call a\alpha,b)=\tau(\alpha,b)$ for $\alpha\in\hist{\Sigma}^*$ and $b\in\Sigma_f$. The play $\Pi(\sigma_{\mathcal A_q},\tau',a)$ is clearly infinite. In the other inductive cases, the second statement follows similarly from the induction hypothesis.
	
\item[``$\Longleftarrow$'':]
Let $w_1\in L(r)$ and a Romeo strategy $\tau_1$ be such that $\Pi(\sigma_{\mathcal A_q},\tau_1,w_1)$ is infinite. The idea is to make use of the fact that there must exist a pair $(q_m,a_m)\in Q\times\Sigma$ such that $\Pi(\sigma_{\mathcal A_q},\tau_1,w_1)$ contains infinitely many configurations with $\mathcal A_q$-state $q_m$ and current symbol $a_m$.

Let $u_1$ be the longest prefix of $w_1$ such that $\Pi(\sigma_{\mathcal A_q},\tau_1,u_1)$ is finite. Define $\alpha_1\in\hist{\Sigma}^*$ by $\alpha_1 \df h(\Pi(\sigma_{\mathcal A_q},\tau_1,u_1))$ and let $a_1\in\Sigma$ and $v_1\in\Sigma^*$ be such that $w_1=u_1a_1v_1$. This means that $\Pi(\sigma_{\mathcal A_q},\tau_1,w_1)$ contains the configuration $(\alpha_1,a_1v_1)$ with $\mathcal A_q$-state $q_1=\delta^*(q,\alpha_1)$ and, by maximality of $u_1$, the play $\Pi(\sigma_{\mathcal A_{q_1}},\tau_1^{\alpha_1},a_1)$ is infinite. Thus, $\delta(q_1,a_1) \in F$ and, letting $\tau_2=\tau_1^{\alpha_1\call{a_1}}$ and $w_{2}=\tau_1(\alpha_1,a_1)$, the play $\Pi(\sigma_{\mathcal A_{\delta(q_1,\call{a_1})}},\tau_2,w_{2})$ must also be infinite. By repeating the argument we obtain infinite sequences $(w_i),(u_i),(v_i)\subseteq\Sigma^*$, $(a_i)\subseteq\Sigma$, $(\alpha_i)\subseteq\hist{\Sigma}^*$, $(q_i)\subseteq Q$ and $(\tau_i)$ of Romeo strategies such that
\begin{itemize}
\item $w_{i+1}=\tau_i(\alpha_i,a_i)=u_{i+1}a_{i+1}v_{i+1}$,
\item $\tau_{i+1}=\tau_i^{\alpha_i\call{a_i}}$,
\item $h(\Pi(\sigma_{\mathcal A_{\delta(q_i,\call{a_i})}},\tau_{i+1},u_{i+1}))=\alpha_{i+1}$,
\item $q_{i+1}=\delta^*(q_i,\call{a_{i}}\alpha_{i+1})$,
\item $\Pi(\sigma_{\mathcal A_{\delta(q_i,\call{a_i})}},\tau_{i+1},w_{i+1})$ is infinite and contains, for each $j>i$, a configuration with $\mathcal A_{\delta(q_i,\call{a_i})}$-state $q_j$ and current symbol $a_j$,
\item $\delta(q_i,a_i) \in F$
\end{itemize}
for each $i\in\mathbb N$.

Since $Q$ and $\Sigma$ are finite but the sequences are infinite, there must exist $m<n$ with $(q_m,a_m)=(q_n,a_n)$. We choose these such that $(n,m)$ is lexicographically minimal and prove the implication by induction on $\abs{r}+\sum_{i=1}^{n-1}\abs{R_{a_i}}$. It is clear that $r\ne\epsilon$ since a play on $\epsilon$ is never infinite. There are several possible cases.

If $r=a\in\Sigma$, then $w_1=a=a_1$, $q_1=q$, $\delta(q,a) \in F$ and $\Pi(\sigma_{\mathcal A_{\delta(q,\call{a})}},\tau_2,w_2)$ is infinite and contains a configuration with $\mathcal A_{\delta(q,\call{a})}$-state $q_n$ and current symbol $a_n$. There are two subcases.
\begin{itemize}
\item If $m=1$ then $q_n=q$ and $a_n=a$. By Claim~\ref{cl:winningSetRegularNext} this means that $(\delta(q,\call{a}),R_{a},q,a)\in\Next$. Thus, condition~\eqref{it:winningSetRegularInfBase} of the definition of $\Inf$ is satisfied and $(q,r)\in\Inf$.
\item If $m>1$ then the induction hypothesis can be applied because the sequences $(a_i)$ and $(q_i)$ for the play $\Pi(\sigma_{\mathcal A_{\delta(q,\call{a})}},\tau_2,w_2)$ are $(a_2,a_3,\dots)$ and $(q_2,q_3,\dots)$. Thus, it holds that $(\delta(q,\call{a}),R_{a})\in\Inf$ and condition~\eqref{it:winningSetRegularInfCall} of the definition yields $(q,r)\in\Inf$.
\end{itemize}

If $r=r_1r_2$, $r=r_1+r_2$ or $r=r_1^*$, then $(q,r)\in\Inf$ follows similarly from the induction hypothesis.
\end{description}
This completes the proof of Claim~\ref{cl:winningSetRegularInf}.

From Claims~\ref{cl:winningSetRegularMove} and \ref{cl:winningSetRegularInf}  it follows that the desired relations can be obtained by specialising the relations $\Move$ and $\Inf$ to single symbol expressions. 
It is not hard to see that the relations $\Move$, $\Next$ and $\Inf$ can be constructed in polynomial time. The inductive definition of these relations can be implemented na\"{\i}vely with several nested loops. The outermost loop terminates when no new tuple was added during its last iteration. Since $Q$, $\Sigma$ and $S$ are all of polynomial size, the entire computation can be executed in polynomial time.
\end{proof}

Now we are prepared to start the proof of Theorem \ref{thm:complexity}. Its upper bounds mostly rely on the following insight on the winning set of regular strategies, which was already stated similarly for forgetful regular strategies in \cite[Proposition 3.1]{AbiteboulMB05}. We prove it here for the sake of completeness and extend the proof to include arbitrary regular strategies and take into account infinite plays.
}% \full

\short{
The upper bounds mostly rely on the following insight on the winning set of regular strategies, which was already stated similarly for forgetful regular strategies in \cite[Proposition 3.1]{AbiteboulMB05}. We extend the proof to include arbitrary regular strategies and take into account non-terminating plays.
}% \short

\begin{proposition}\label{prop:winningSetRegular}
  \begin{enumerate}[(a)]%[label=(\alph*)]
  \item For a game $G=(\Sigma,R,T)$ and a strategy automaton
    $\mathcal A$, the set $W(\sigma_{\mathcal A})$ is regular.
  \item  
  An NFA $\mathcal N$ recognising $\Sigma^*\setminus W(\sigma_{\mathcal A})$ can be
    constructed from $G$ and $\mathcal{A}$ in polynomial time.
  \end{enumerate}
\end{proposition}

\short{
The lower bounds in Theorem~\ref{thm:complexity} are shown by reductions from \textsc{3SAT} and the universality problem for NFAs. Complete proofs are given in the full version of the paper.
}

\full{
\begin{proof}
Clearly (a) follows from (b) and therefore it suffices to show (b).

The idea behind the construction of the NFA $\mathcal{N}$ recognising $\Sigma^*\setminus W(\sigma_{\mathcal A})$ is basically that the nondeterminism in $\mathcal{N}$ will be used to guess a counter-strategy (consisting of choosing replacement words with associated sub-strategies) to the fixed strategy $\sigma_{\mathcal{A}}$. The groundwork for this was already laid in the form of the $\Move$ relation from the proof of Proposition \ref{prop:strategyAutomatonTests}. However, the target language $L(T)$ and the relation $\Inf$ from the proof of Proposition \ref{prop:strategyAutomatonTests} need to be taken into account, which requires some additional constructions.

First, we consider the product automaton $\mathcal{A}'$ of $\mathcal{A}$ with $T$. Let $\mathcal{A} = (Q_{\mathcal{A}}, \hist{\Sigma}, \delta_{\mathcal{A}}, s_{\mathcal{A}},F_{\mathcal{A}})$ and $T = (Q_T, \Sigma, \delta_T, s_T, F_T)$. We define $\mathcal{A}' = (Q_{\mathcal{A}'},\hist{\Sigma},\delta_{\mathcal{A}'},s_{\mathcal{A}'},F_{\mathcal{A}'})$, where $Q_{\mathcal{A}'}=Q_{\mathcal{A}} \times Q_T$ and $s_{\mathcal{A}'} = (s_{\mathcal{A}}, s_T)$ as per the standard product construction. Since $\mathcal{A}'$ will take the place of a strategy automaton in the construction from the proof of Proposition \ref{prop:strategyAutomatonTests}, we set its accepting states to $F_{\mathcal{A}'} = F_\mathcal{A} \times Q_T$ and define its transition function $\delta_{\mathcal{A}'}$ as follows for symbols $a \in \Sigma$, $\call{a} \in \call{\Sigma}$:
\begin{align*}
\delta_{\mathcal A'}((p,q),a)&= (\delta_{\mathcal A}(p,a),\delta_T(q,a))\text{, and}
\\
\delta_{\mathcal A'}((p,q),\call a)&=(\delta_{\mathcal A}(p,\call a),q).
\end{align*}

It is easy to verify that $\mathcal{A}'$ defines a strategy with the same winning set as $\sigma_{\mathcal{A}}$, and that, on an input $\alpha \in \hist{\Sigma}$, its second component simulates the DFA $T$ on $\natural \alpha$. Clearly, $\mathcal{A}'$ can also be computed from $G$ and $\mathcal{A}$ in polynomial time. Since $\mathcal{A}'$ is a strategy automaton of polynomial size, it is also possible to compute its $\Move$ and $\Inf$ relations from the proof of Proposition \ref{prop:strategyAutomatonTests} in polynomial time.

We now define the NFA $\mathcal N$ for $\Sigma^*\setminus W(\sigma_{\mathcal A})$ by $\mathcal N=(Q_{\mathcal{A}'}\cup\{q_\infty\},\Sigma,\delta_{\mathcal N},s_{\mathcal{A}'},F_{\mathcal{N}})$, where $F_{\mathcal{N}} = (Q_\mathcal{A} \times (Q_T \setminus F_T)) \cup\{q_\infty\}$ and the transition relation $\delta_{\mathcal N}$ is defined as the union of the sets 
$\delta_{\mathcal N} = \{((p,q),a,(p',q'))\in\Move\mid a\in\Sigma\}$, $\{((p,q),a,q_\infty) \mid (p,a)\in\Inf, a\in\Sigma\}$ and $\{(q_\infty,a,q_\infty)\mid a\in \Sigma\}$. We note that, whenever the $\mathcal A$-component reaches a situation in which Romeo can enforce an infinite play, $\mathcal N$ can enter the accepting state $q_\infty$ and remain there.
Clearly, $\mathcal N$ can be constructed in polynomial time. It thus only remains to show that $L(\mathcal N)=\Sigma^*\setminus W(\sigma_{\mathcal A})$. We prove both inclusions separately.

\begin{description}%[listparindent=\parindent]
\item[``$\subseteq$'':] Suppose $w=a_1\dots a_n\in L(\mathcal N)$ with $a_i\in\Sigma$. Let us assume first that there is an accepting run of $\calN$ on $w$ which ends in a state of $Q_\mathcal{A} \times (Q_T \setminus F_T)$. Then there exists a sequence $(p_0,q_0),\dots,(p_n,q_n)\in Q_{\mathcal{A}'}$ such that $(p_0,q_0)=(s_{\mathcal{A}},s_T)$, $q_n\in Q_T\setminus F_T$, and $((p_{i-1},q_{i-1}),a_i,(p_i,q_i))\in\delta_{\mathcal N}$ for each $i=1,\dots,n$. By construction, it holds for each $i=1,\dots,n$ that $((p_{i-1},q_{i-1}),a_i,(p_i,q_i))\in\Move$, and by Claim~\ref{cl:winningSetRegularMove} as well as the above considerations, there exist Romeo strategies $\tau_1,\dots,\tau_n$ such that $\Pi(\sigma_{\mathcal A_{p_{i-1}}},\tau_i,a_i)$ terminates with a final history string $\alpha_i$ such that $\delta^*_{\mathcal{A}}(p_{i-1},\alpha_i) = p_i$ and $\delta_T^*(q_{i-1}, \natural \alpha_i) = q_i$
for each $i=1,\dots,n$. Then the play of $\sigma_{\mathcal A}$ against the Romeo strategy that plays like $\tau_i$ on the $i$th symbol of the input word $w$ yields a final history string $\alpha = \alpha_1 \cdots \alpha_n$ with $\delta_T^*(s_T, \natural \alpha) = q_n \in Q_T \setminus F_T$.

In the other case there is an accepting run that ends in $q_\infty$. Let $(p,q)$ be the last state of $\calN$ in this run before $q_\infty$ is entered by reading a symbol $a$. By construction of $\calN$, $(p,a)\in\Inf$. Similarly as in the first case there exists a strategy of Romeo leading into a configuration in which  $\calA$ has state $p$ with $a$ as the next symbol. Romeo can enforce an infinite play from here and therefore $w\not\in W(\sigma_{\mathcal A})$.

\item[``$\supseteq$'':] Suppose $w\in \Sigma^*\setminus W(\sigma_{\mathcal A})$ for some $w=a_1\dots a_n$ with $a_i\in\Sigma$. Then there exists a Romeo strategy $\tau$ such that $\Pi(\sigma_{\mathcal A},\tau,w)$ is infinite or losing for Juliet. We first consider the case that it is losing for Juliet. For $i=1,\dots,n$, let $(\alpha_i,a_{i+1}\dots a_n)$ be the configuration at the time when $a_1\dots a_i$ is completely processed, let $p_i = \delta_{\mathcal{A}}^*(s_{\mathcal{A}},\alpha_i)$ be its $\mathcal A$-state and $q_i = \delta_T^*(s_T,\natural \alpha_i)$ its $T$-state, and let $\alpha_0=\epsilon$ (the history string of the initial configuration). Then, for each $i=0,\dots,n-1$, $h(\Pi(\sigma_{\mathcal A_{p_{i}}},\tau^{\alpha_{i}},a_{i+1}))=\beta_i$ such that $\alpha_{i+1} = \alpha_i \beta_i$, and therefore $\delta_{\mathcal{A}}(p_i, \beta_i) = p_{i+1}$ and $\delta_T(q_i, \natural \beta_i) = q_{i+1}$. Thus, by Claim~\ref{cl:winningSetRegularMove} and the above considerations, $((p_i,q_i),a_{i+1},(p_{i+1},q_{i+1}))\in\Move$ for each $i=0,\dots,n-1$, so there exists a run of $\mathcal N$ on $w$ that terminates in $(p_n,q_n)$. Since $\Pi(\sigma_{\mathcal A},\tau,w)$ is losing for Juliet, it holds that $q_n = \delta_T^*(s_T,\natural \alpha_n)\in Q_T\setminus F$, so $w\in L(\mathcal N)$.

In the other case, when $\Pi(\sigma_{\mathcal A},\tau,w)$ is infinite, let $m$ be maximal such that the play on $a_1\dots a_m$ is finite. Let the states $p_i$ and $q_i$ be defined as before, for $i=1,\ldots,m$. Since the play on the next symbol is infinite, $(p_m,a_{m+1})\in\Inf$, and therefore $\calN$ has a run that reaches $(p_m,q_m)$ and then enters $q_\infty$ and remains there, so $w\in L(\mathcal N)$. 
\end{description}
\vspace{-\baselineskip}
\end{proof}

Now we are ready to prove the upper bounds of Theorem \ref{thm:complexity}. 

\begin{proposition}\label{prop:complexity-upper}
  \begin{enumerate}[(a)]%[label=(\alph*)]
  \item $\textsc{IsWinning} \in \P$. 
  \item $\textsc{ExistsWinning($\tFive$)}\in\NP$.
  \item $\textsc{IsDominated} \in \PSPACE$.
  \end{enumerate}
\end{proposition}

\begin{proof}
  \begin{enumerate}[(a)]%[label=(\alph*)]
  \item 
 	By Proposition~\ref{prop:winningSetRegular}, an NFA $\mathcal N$ recognizing $\Sigma^*\setminus W(\sigma_{\mathcal A})$ can be constructed in polynomial time from $\mathcal A$ and $G$. Then, the set of states of $\mathcal N$ that can be reached at the end of a run on $w$ can be computed in polynomial time by simulating the power set automaton. We have $w\in W(\sigma_{\mathcal A})$ if and only if none of these states is accepting.
      \item \textsc{ExistsWinning($\tFive$)} can be solved non-deterministically in polynomial time as follows. Let $G=(\Sigma,R,T)$ and $w\in\Sigma^*$ be the input.
      First, guess a strongly regular strategy automaton $\mathcal A$ by adding a state $\Call$ to $T$ and nondeterministically guessing which transitions of $T$ to reroute to $\Call$. Second, call the algorithm for \textsc{IsWinning} from (a) with input $G$, $\mathcal A$ and $w$. 
      \item Let $G$ be a game and let $\mathcal A_1$ and $\mathcal A_2$ be strategy automata for $G$. By Proposition~\ref{prop:winningSetRegular}, NFAs $\mathcal B_1$ and $\mathcal B_2$ recognising $\Sigma^*\setminus W(\sigma_{\mathcal A_1})$ and $\Sigma^*\setminus W(\sigma_{\mathcal A_2})$, respectively, can be constructed in polynomial time.
Clearly, $W(\sigma_{\mathcal A_1})\subseteq W(\sigma_{\mathcal A_2})$ if and only if 
$L(\mathcal B_2)\subseteq L(\mathcal B_1)$. Therefore, the relative domination can be decided with the help of the polynomial space algorithm that tests containment for NFAs  \cite{HuntRS76}.
 \end{enumerate}
\end{proof}

We now turn to the corresponding lower bound results for Theorem~\ref{thm:complexity}.

\begin{proposition}\label{prop:existswinning-lower}
	The problem \textsc{ExistsWinning($\tFive$)} is $\NP$-hard, even for non-recursive games whose replacement languages are finite and prefix-free.
\end{proposition}

\begin{proof}

We give a reduction from the $\NP$-complete problem \textsc{3SAT} of determining whether a given Boolean formula in CNF with three literals per clause is satisfiable. To this end, let $\varphi$ be an arbitrary formula in 3CNF over variables $x_1, \ldots, x_n$ with clauses $K_1, \ldots, K_m$. We describe how to construct from $\varphi$ a game $G = (\Sigma,R,T)$ and string $w \in \Sigma^*$ in polynomial time, such that Juliet has a strongly regular winning strategy in $G$ on $w$ if and only if $\varphi$ is satisfiable.

The basic idea behind this reduction is to encode variable assignments into strategy automata of a strongly regular strategy. To this end,  the input string $w$ is of the form $w_1w_2$, where the sub-game on $w_1$ is meant to choose a variable assignment and the sub-game on $w_2$ to verify that it is satisfying. We consider only certain \emph{valid} strategies $\sigma$ for Juliet, each of which will induce a variable assignment $\theta_\sigma$ for $\varphi$.

The string $w_1$ is chosen as $w_1 = 0CDa_{2}0CDa_{3}0CD\cdots a_{n}0CD$.

The DFA $T = (Q, \Sigma, \delta,s,F)$ consists of sub-automata $T_1, \ldots, T_n$, where for each $i\le n$, $T_i$ has states and transitions as specified by Figure~\ref{fig:NPcompleteTj}. Furthermore, $Q$ contains a non-accepting sink state $q_e$. The set of accepting states is $F=\{f_{1},\dots,f_n\}$, and the initial state is $s = s_1$. The construction of $T$ is completed by adding the following transitions:
\begin{itemize}
	\item $(t_{i-1},a_i,s_i)$ for $i=2,\dots,n$ (connecting $T_{i-1}$ and $T_i$),
	\item $(q,a_i,s_i)$ for each $q\in Q\setminus(F\cup \{q_e,t_{1},t_{2},\dots,t_{n-1}\})$ and $i=1,\dots,n$, and 	\item transitions leading to the error state $q_e$ for all remaining combinations of states and symbols.
\end{itemize}

\begin{figure}
\centering
\begin{tikzpicture}[xscale=2.3,yscale=2.3,>=stealth',initial text=,every state/.style={inner sep=0pt,minimum size=1.5em}]
\node[state] at (0,0) (j) {$s_i$};
\node[state] at (-1,0) (0) {$b_i$};
\node[state] at (-1,1) (C) {$c_i$};
\node[state,accepting] at (1,1) (fj) {$f_i$};
\node[state] at (1,0) (1) {$d_i$};
\node[state] at (2,0) (lj) {$t_i$};
\path[->] (j)   edge node[below] {$0$} (0)
                edge[bend left]  node[above] {$b$}   (fj)
                edge node[below] {$1$}   (1)
          (0)	edge node[left] {$C$}   (C)
          (C)	edge node[above] {$d$}   (j)
          (1)	edge node[below] {$D$}   (lj)
          (fj)  edge node[right] {$c$}   (j)
				edge[out=135,in=45,loop] node[above] {$b,a_{1},\dots,a_n$} ()
                edge             node[above] {$d$}   (lj);
\end{tikzpicture}
\caption{Component $T_i$ of the automaton $T$ in the proof of Theorem~\ref{thm:complexity}~(\ref{thm:NPcomplete})}\label{fig:NPcompleteTj}
\end{figure}
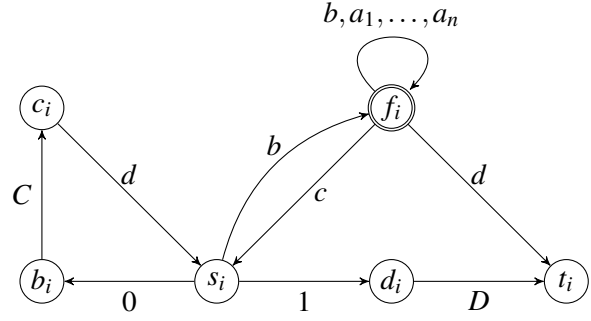
Roughly speaking, every sub-string $0CD$ of $w_1$ induces a sub-game whose strategy is determined by one sub-automaton $T_i$ and the symbols $a_i$ are used to get from $T_{i-1}$ to $T_i$.

Before we further describe the construction of $G$, we first show that $T$ is a minimal automaton. This is because all states are reachable from the initial state $s_1$ and the states are pairwise non-equivalent:
\begin{itemize}
\item The sink state is not equivalent to any other state because it is the only state from which no accepting state is reachable.
\item For any two states $p$ and $q$ with distinct roles within their components ($p$ and $q$ may belong to the same component $T_i$ or to two distinct components $T_j$ and $T_k$ with $j\ne k$) it is easy to find some symbol $e\in\Sigma$ such that $\delta(p,e)=q_e$ and $\delta(q,e)\ne q_e$ or vice versa.
\item Two states $q_j$ and $q_k$ with the same role in their respective components $T_j$ and $T_{k}$ with $j<k$ are not equivalent: Let $u\in\{1,c,C,d,D,a_{k+1},\dots,a_{n}\}^*$ be of minimal length such that $\delta^*(q_k,u)=t_n$. Then $\delta^*(q_k,ua_{1})=s_1$ and $\delta^*(q_j,ua_{1})=q_e$.
\end{itemize}
Thus, the strongly regular strategy automata for $G$ are exactly the automata that can be obtained from $T$ by adding a state $\Call$ as the only accepting state, removing transitions  with function symbols from $T$ and adding new transitions with function symbols to $\Call$ instead.

A strategy automaton for Juliet derived from $T$ is called \emph{valid} if, for every $i$, $T_i$  has exactly one of the two transitions $(s_i,0,b_i)$ and $(s_i,1,d_i)$. A strategy is called valid if it is defined by a valid strategy automaton. With a valid strategy $\sigma$ we associate a variable assignment $\theta_\sigma$ as follows: $\theta_\sigma(x_i)=0$ if and only if the transition $(s_i,0,b_i)$ is \emph{missing} (and $\theta_\sigma(x_i)=1$ if the transition $(s_i,1,d_i)$ is missing). 

For the sub-game on $w_1$, the set $R$ has the following replacement rules.
\begin{alignat*}{2}
&0&&\to b\\
&1&&\to b\\
&C&&\to c1\\
&D&&\to d1d
\end{alignat*}
The symbols $b,c,d$ and $a_1,\ldots,a_n$ are \emph{not} function symbols.

These rules ensure that any strongly regular strategy $\sigma$ with strategy automaton $\mathcal{A}$ for Juliet  rewrites $w_1 = 0CDa_{2}0CDa_{3}0CD\cdots a_{n}0CD$ into a string $w_1'$ with $\delta^*(s_1,w_1') \neq q_e$ only if it is valid. It should be noted that this statement is independent of any strategy of Romeo, since the above four rules do not leave Romeo any choice.

The statement follows by an easy induction on prefixes of $w_1$ of the form $0CDa_{2}0\cdots a_i0CD$. 
 
The crucial observation is that a sub-play on a subword $0CD$ starting from a state $s_i$ reaches the sink state $q_e$ if either both or none of $(s_i,0,b_i)$ and $(s_i,1,d_i)$ are present. And there is no way out of $q_e$ afterwards. 

Indeed, if both transitions are present, the sub-play proceeds as follows: it reads 0. If the transition $(b_i,C,c_i)$ is missing, $C$ is replaced by $c1$ yielding the sink state after reading $c$. Otherwise, $C$ is read, and the subsequent $D$ either leads to $q_e$ or is necessarily replaced by $d1d$. In the latter case, $d$ is then read, bringing the play back to $s_i$. By assumption, $1$ is read as well and now $d$ yields a transition into $q_e$. 

On the other hand, if both transitions are missing, the sub-play proceeds as follows: it calls 0, getting $b$ which is read. It then calls $C$, getting $c1$. Reading $c$ leads back into $s_1$,where the subsequent 1 triggers another call and the resulting $b$ leads back to $f_i$. The call of $D$ yields $d1d$, reading $d$ leads into $t_i$. Either $q_e$ is now immediately reached by reading 1 or after calling it and reading $b$.

Altogether, we have established that any strongly regular winning strategy for Juliet on $w$ must be valid.
We are ready to complete the construction of $R$ and finish the proof. 

To this end, for each $j\le m$, let $K_j = (\ell_{j,1} \vee \ell_{j,2} \vee \ell_{j,3})$ be a clause of $\varphi$ with literals $\ell_{j,1}$, $\ell_{j,2}$, $\ell_{j,3}$. By $w(K_j)$ we denote the string $w(\ell_{j,1})w(\ell_{j,2})w(\ell_{j,3})$, where for an arbitrary literal $\ell$, we define

	$$w(\ell) = \begin{cases}
                    a_i0, & \text{if $\ell = \neg x_i$ for some $i$}\\
                    a_i1, & \text{if $\ell = x_i$ for some $i$.}\\
                \end{cases}$$
                
Finally, $w_2 = E$, where $E$ is a function symbol with replacement rules
	$$E \to w(K_1) \mid w(K_2) \mid \ldots \mid w(K_m).$$
It is not hard to see that this construction is indeed feasible in polynomial time and that $G$ is non-recursive with finite prefix-free replacement languages. It remains to be proven that Juliet has a strongly regular winning strategy on $w$ in $G$ if and only if $\varphi$ is satisfiable.

For the ``if'' direction, let $\theta$ be a satisfying assignment for $\varphi$ and consider the strongly regular strategy $\sigma_\theta$ with strategy automaton $\mathcal{A}_\theta$ derived from $T$ by adding the accepting $\Call$ state and rerouting to it from each component $T_i$ of $T$ all transitions leading to $q_e$ and, furthermore, the transition $(s_i,0,b_i)$ if $\theta(x_i)=0$ and the transition $(s_i,1,d_i)$, otherwise. We claim that $\sigma_\theta$ is winning on $w$.

It is easy to see, that every sub-play on $0CD$ starting in a state $s_i$ according to $\sigma_\theta$ leads to state $t_i$. Either through reading $0C$, calling $D$, reading $d$, calling 1, reading $b$ and $d$. Or through calling 0, reading $b$, calling $C$, and reading $c1D$.

Thus, with the help of the additional $a_i$-symbols, the sub-play on $w_1$, which starts in $s_1$, ends necessarily in $t_n$. By construction of $\sigma_\theta$, Juliet then calls the final $E$ of $w$. Let $j$ be such that $w(K_j)$ is Romeo's answer to Juliet's Call on $E$. Since $\theta$ satisfies $\varphi$, at least one literal of $K_j$ is set to ``true'' under $\theta$, which by the construction of $\mathcal{A}_\theta$ and the definition of $w(K_j)$ means that $w(K_j)$ contains at least one substring of the form $a_i0$ (resp. $a_i1$) for which the $0$-transition (resp. $1$-transition) is missing from $T_i$ in $\mathcal{A}_\theta$. It is easy to see that, after reading $a_i$, $T$ is in state $s_i$, so the subsequent $0$ (resp. $1$) is called by Juliet and thus rewritten into $b$, inducing a transition into $f_i$. Since $f_i$ has no outgoing $0$- or $1$-transitions in $\mathcal{A}_\theta$ and the remainder of $w(K_j)$ contains only symbols from $\{0,1, a_1, \ldots, a_n\}$, it is clear that all remaining symbols $0$ or $1$ get rewritten to $b$, and $T$ never leaves $f_i$ anymore. This implies that $\sigma_\theta$ is winning on $w$.

For the ``only if'' direction, assume that $\sigma=\sigma_{\mathcal A}$ is a strongly regular winning strategy for Juliet on $w$ induced by a strategy automaton $\mathcal{A}_\sigma$. 
As shown above, $\sigma$ must be valid and thus induces an assignment $\theta_\sigma$ for the variables of $\varphi$. We need to show that $\theta_\sigma$ satisfies at least one literal in each clause of $\varphi$, and is therefore a satisfying assignment for $\varphi$.

We first prove that the assumption that $\sigma$ is a winning strategy on $w$ implies that the sub-play on $w_1$ must end in state $t_n$. To show this, we consider again a sub-play on a string $0CD$ starting from a state $s_i$. We need to consider two cases.

In the first case, where $(s_i,0,b_i)$ is present, $(b_i,C,c_i)$ must also be present since otherwise, reading 0 and calling $C$ would yield the next symbol $c$, which can not be called and brings $T$ from $b_i$ to $q_e$, contradicting our winning assumption. Thus, $C$ is read, $D$ is called, $d$ is read, $1$ is called and finally $bd$ is read, reaching $t_i$.

In the other case, where $(s_i,0,b_i)$ is absent, $0$ is called, $b$ is read, $C$ is called, and $c1$ is read. If the transition $(d_i,D,t_i)$ were absent, the subsequent call of $D$ would yield $d1d$ and reading $d$ would bring $T$ into state $q_e$, again a contradiction. Thus, $D$ is read and the sub-play ends in $t_i$, as desired.

By induction and, taking the intermediate non-function symbols $a_i$ into account, the play on $w_1$ thus indeed ends in $t_n$. 

In the remaining play on $w_2=E$, $\sigma$ has to call $E$, since $\delta(t_n,E)=q_e$. Each counter-strategy of Romeo consists of answering this call by some $w(K_j)$. We show that, if Juliet wins according to $\sigma$ against a counter-strategy playing $w(K_j)$, then the assignment $\theta_\sigma$ satisfies clause $K_j$; since $\sigma$ is a winning strategy, this directly implies that $\theta_\sigma$ is a satisfying assignment for $\varphi$.

To this end, assume for the sake of contradiction that $\theta_\sigma$ does \emph{not} satisfy $K_j$. This means that $\theta_\sigma(\ell_{j,k})=0$, for every $k\in\{1,2,3\}$. Thus, for each substring of the form $a_i0$ (resp. $a_i1$) of $w(K_j)$, the $0$-transition (resp. $1$-transition) of $T_i$ is present in $\mathcal{A}_\sigma$. It is easy to see that, starting from $t_n$, the sub-play on $w(K_j)$ consists of six reading steps, and eventually ends in a non-accepting state of the form $s_k$. Therefore $\sigma$ does \emph{not} win on $w$ against the counter-strategy of Romeo replacing $E$ with $w(K_j)$. This contradicts the assumption that $\sigma$ is winning on $w$ and completes the proof.
\end{proof}

\begin{proposition}\label{prop:compare-lower}
The problem \textsc{IsDominated} is $\PSPACE$-hard, even for strongly regular strategies and for non-recursive games 
whose plays all have depth at most $2$, and whose replacement words all have length $1$.
\end{proposition}

\begin{proof} 
The proof is by a reduction from the $\PSPACE$-complete universality problem for NFAs \cite{MeyerS72,HuntRS76}:  Given an NFA $\mathcal N$ over the alphabet $\{0,1\}$, does $L(\mathcal N)=\{0,1\}^*$?

The idea is to construct $G$ such that there is a strongly regular strategy $\sigma_1$ that loses on a string from $\{0,1\}^*$ if and only if it is accepted by $\mathcal N$ and a second one, $\sigma_2$, that loses exactly on all words from $\{0,1\}^*$. Then $W(\sigma_1)\subseteq W(\sigma_2)$ if and only if $L(\mathcal N)=\{0,1\}^*$.

\newcommand{\QN}{\ensuremath{Q_{\mathcal N}}}

Let thus $\mathcal N=(Q_{\mathcal N}, \{0,1\}, \delta_{\mathcal N},s,F_{\mathcal N})$ be an NFA. We construct a game $G=(\Sigma,R,T)$ and strongly regular strategy automata $\mathcal A_1$, $\mathcal A_2$ for $G$ such that $L(\mathcal N)=\{0,1\}^*$ holds if and only if $W(\sigma_{\mathcal A_1})\subseteq W(\sigma_{\mathcal A_2})$. 

Without loss of generality, we assume that $s\in F_{\mathcal N}$. Otherwise $\mathcal N$ is not universal due to $\epsilon\notin L(\mathcal N)$ and $\mathcal N$ can be mapped to a fixed negative instance of \textsc{IsDominated}. We also assume without loss of generality that all states of $\mathcal N$ are reachable from the initial state $s$ (otherwise, unreachable states can be removed).

We construct $G=(\Sigma,R,T)$ with $\{0,1\}\subset\Sigma$ and automata $\mathcal A_1$ and $\mathcal A_2$ such that
\begin{enumerate}[(a)]
\item $L(\mathcal N)=\{0,1\}^*\setminus W(\sigma_{\mathcal A_1})$, \ie $\mathcal N$ recognizes the language of words in $\{0,1\}^*$ on which $\sigma_{\mathcal A_1}$ does not win, and \label{it:compareA1}
\item $W(\sigma_{\mathcal A_2}) = \Sigma^*\setminus \{0,1\}^*$.\label{it:compareA2} \end{enumerate}
Then it holds that $L(\mathcal N)=\{0,1\}^*$ if and only if $W(\sigma_{\mathcal A_1})\cap\{0,1\}^*=\emptyset$, which in turn holds if and only if $W(\sigma_{\mathcal A_1})\subseteq W(\sigma_{\mathcal A_2})$, since the latter equals $\Sigma^*\setminus\{0,1\}^*$. Thus it suffices to show that $G$, $\mathcal A_1$ and $\mathcal A_2$ with \eqref{it:compareA1} and \eqref{it:compareA2}  can be constructed in polynomial time. An example of the subsequently described construction is shown in Figure~\ref{fig:compare}.

We define the alphabet of $G$ as
$\Sigma=\{0,1,\#\}\cup\{(a,p)\mid a\in\{0,1\}, p\in \QN\}\cup\{(\$,p)\mid p\in Q_{\mathcal N}\setminus\{s\}\}.$ 

The set $R$ of rules of $G$ is given by
\begin{align*}
 a &\to \bigvee_{p\in  Q_{\mathcal N}} (a,p) &&\text{for }a\in\{0,1\}\\
 (a,p)&\to \# &&\text{for $a\in\{0,1\}$ and $p\in \QN$}\\
 (\$,p)&\to \# &&\text{for }p\in Q_{\mathcal N}\setminus\{s\}.
\end{align*}
We note that $G$ is non-recursive, all possible plays have depth at most $2$ and all replacement words are of length $1$.

The target language automaton of $G$ is $T=(Q,\Sigma,\delta,s,F)$, where $Q=Q_{\mathcal N}\dotCup\{f\}$ for some new state $f\notin Q_{\mathcal N}$, the set $F=\{f\}\dotCup Q_{\mathcal N}\setminus F_{\mathcal N}$ of accepting states consists of $f$ and the non-accepting states of $\mathcal N$, and the transition function $\delta$ is defined as follows, for $a\in\{0,1\}$ and $p,q\in\QN$:
\begin{itemize}
\item $\delta(p,a)=p$,
\item $\delta(p,(a,q))=
  \begin{cases}
  q &\text{if $(p,a,q)\in \delta_{\mathcal N}$,}\\
  f & \text{otherwise,}
  \end{cases}$
\item $\delta(p,(\$,p))=s$, if $p\not=s$,
\item all remaining transitions lead to the state $f$.
\end{itemize}

\begin{figure}
\centering
\subcaptionbox{NFA $\mathcal N$}[.28\textwidth]{
	\centering
	\begin{tikzpicture}[xscale=2.3,yscale=2.3,>=stealth',initial text=,initial distance=1.5ex,every state/.style={inner sep=0pt,minimum size=1.5em}]
	\node[state,initial,accepting] at (0.1,0.5) (q0) {$s$};
	\node[state,accepting] at (1,1) (q1) {$b$};
	\node[state] at (1,0) (q2) {$c$};
	\path[->] (q0)	edge[bend left] node[above] {$0$} (q1)
	edge[bend right] node[below] {$0$} (q2)
	(q1)	edge[out=130,in=50,loop]	node[above] {$0$} ()
	edge				node[right] {$1$} (q2);
	\end{tikzpicture}}% Do not remove this comment as otherwise it will cause a spurious blank space to be added, the total length will surpass \textwidth and the figures will end up not side-by-side
\subcaptionbox{Target language DFA $T$. \newline All transitions that are not shown end in $f$.}[.450\textwidth]{
	\centering
	\begin{tikzpicture}[xscale=2.3,yscale=2.3,>=stealth',initial text=,initial distance=1.5ex,every state/.style={inner sep=0pt,minimum size=1.5em}]
	%\clip (2.95,-1.1) rectangle (6.2,2.7);
	\node[state,initial] at (3.5,0.5) (q0') {$s$};
	\node[state] at (4.4,1) (q1') {$b$};
	\node[state,accepting] at (4.4,-0.0) (q2') {$c$};
	\node[state,accepting] at (5.4,0.5) (qf) {$f$};
	\path[->] (q0')	edge[out=-150,in=-70,loop]	node[below] {$0,1$} ()
					edge[bend left] node[above, sloped] {$(0,b)$} (q1')
					edge[bend right] node[below, sloped] {$(0,c)$} (q2')
	(q1')	edge[out=130,in=50,loop] node[above, sloped] {$0,1,(0,b)$} ()
			edge[bend left] node[above, sloped] {$(\$,b)$} (q0')
			edge node[right] {$(1,c)$} (q2')
%			edge node[above] {$\quad\;\;\;\; 0_2,\$_{s},\$_{c}$} (qf)
	(q2')	edge[out=-145,in=-65,loop]	node[below] {$0,1$} ()
			edge[bend right] node[below, sloped] {$(\$,c)$} (q0');
	\end{tikzpicture}}% Do not remove this comment as otherwise it will cause a spurious blank space to be added, the total length will surpass \textwidth and the figures will end up not side-by-side
\subcaptionbox{Rule set $R$.  The symbol $*$ represents all states in \QN.}[.27\textwidth]{
	\centering
	\begin{tikzpicture}[xscale=2.3,yscale=2.3,>=stealth',initial text=,initial distance=1.5ex,every state/.style={inner sep=0pt,minimum size=1.5em}]
	\node[anchor=north west,align=left] at (1.53,1.8) {
		\begin{minipage}{3.5cm}
		\begin{align*}
		 0&\to (0,s)+(0,b)+(0,c)\\
		 1&\to (1,s)+(1,b)+(1,c)\\
		 (\$,b)&\to \#\\
		 (\$,c)&\to \#\\
                 (0,*)&\to \#\\
                 (1,*)&\to \#
		\end{align*}
		\end{minipage}
	};
	\end{tikzpicture}}

\subcaptionbox{Strategy automaton $\mathcal A_1$. Here, $*$ represents all states from $\QN$. All transitions not shown here end in the new accepting state $\Call_1$.}[.480\textwidth]{
	\centering
	\scalebox{0.9}{
	\begin{tikzpicture}[xscale=2.3,yscale=2.3,>=stealth',initial text=,initial distance=1.5ex,every state/.style={inner sep=0pt,minimum size=1.5em}]
	\node[state,initial] at (3.5,0.5) (q0') {$s$};
	\node[state] at (4.4,1) (q1') {$b$};
	\node[state] at (4.4,-0.0) (q2') {$c$};
	\node[state] at (5.9,0.5) (qf) {$f$};
	\path[->] (q0')	edge[bend left] node[above, sloped] {$(0,b)$} (q1')
	edge[bend right] node[below, sloped] {$(0,c)$} (q2')
	(q1')	edge[out=130,in=50,loop] node[above, sloped] {$(0,b)$} ()
	edge node[right] {$(1,c)$} (q2')
	edge node[above, sloped] {\scriptsize$(0,s),(0,c),(1,s),(1,b)$} (qf)
	(q2') edge node[below, sloped] {$(0,*), (1,*)$} (qf)
	(qf)	edge[out=-110,in=-30,loop]	node[below] {$(0,s),(1,*)$} ();
	\draw[->] (q0') .. controls ++(0,1.9) and ++(-0.0,2.1) .. node[above] {$(0,s),(1,*)$} (qf);
	\end{tikzpicture}}}
	
\hspace{3mm}%
\subcaptionbox{Strategy automaton $\mathcal A_2$. All transitions not shown here end in the new accepting state $\Call_2$. }[.480\textwidth]{
	\centering
	\begin{tikzpicture}[xscale=2.3,yscale=2.3,>=stealth',initial text=,initial distance=1.5ex,every state/.style={inner sep=0pt,minimum size=1.5em}]
	\node[state,initial] at (3.5,0.5) (q0') {$s$};
	\node[state] at (4.4,1) (q1') {$b$};
	\node[state] at (4.4,-0.0) (q2') {$c$};
	\node[state] at (5.4,0.5) (qf) {$f$};
	\path[->] (q0')	edge[out=-150,in=-70,loop]	node[below] {$0,1$} ()
	(qf)	edge[out=-110,in=-30,loop]	node[below] {$\Sigma$} ();
	\draw[->] (q0') .. controls ++(0,1.3) and ++(-0.0,1.5) .. node[above] {$\#$} (qf);
	\end{tikzpicture}}	
\caption{Example of an NFA, and the rules, target language DFA, and strategy automata constructed from it, as in the proof of Proposition~\ref{prop:compare-lower}}\label{fig:compare}
\end{figure} 

The automaton $T$ is in fact the minimal DFA for $L(T)$ because its states are reachable (by assumption on $\mathcal N$) and pairwise non-equivalent:
The two states $s$ and $f$ are non-equivalent because $f$ is accepting and $s$ is non-accepting (since $s$ is accepting in $\mathcal N$ by assumption). Each remaining state $q\in Q_{\mathcal N}\setminus\{s\}$ is not equivalent to any other state of $T$ because it is the only state from which reading $(\$,q)$ leads to a non-accepting state (namely $s$). This is actually the only purpose of the symbols of the form $(\$,q)$.

Since $T$ is minimal, strongly regular strategy automata $\mathcal A_1$ and $\mathcal A_2$ can be derived from $T$ by adding states $\Call_1$ for $\mathcal{A}_1$ and $\Call_2$ for $\mathcal{A}_2$, deleting some transitions and replacing them by transitions to the respective $\Call$ state. 

To achieve property~\eqref{it:compareA1}, the idea is that each run of $\mathcal N$ on an input word $w\in\{0,1\}^*$ should correspond to a play, in which Romeo chooses the  replacement words according to the choices in the run. Therefore, Romeo shall have a winning strategy on $w$ if and only if there exists an accepting run of $\mathcal N$ on $w$.

To this end, we define $\mathcal A_1$ as the automaton obtained from $T$ by keeping only the $(a,p)$-edges (for $a\in\{0,1\}$ and $p\in \QN$) and rerouting all others to state $\Call_1$. 

\textit{Claim:}  $\mathcal A_1$ satisfies \eqref{it:compareA1}:\\
It is easy to show by induction on the length of an input word $w\in\{0,1\}^*$ that $\mathcal N$ has a run ending in some state $q\in\QN$ if and only if Romeo has a strategy for which the play against $\sigma_{\mathcal A_1}$ yields a final word $v$ with $\delta^*(s,v)=q$.  Therefore, Romeo has a winning strategy against $\sigma_{\mathcal A_1}$ precisely on all words in $L(\mathcal N)$, since each run ending in an accepting state $q$ of $\mathcal N$ corresponds to a play that yields the non-accepting state $q$ of $T$ and vice versa. It is important to note here, that plays in which Romeo chooses a rewriting to a symbol $(a,p)$, for which no corresponding transition in $ \mathcal N$ exists, end in the accepting state $f$ of $T$ and are therefore losing for Romeo. This completes the proof of the claim.

The automaton $\mathcal A_2$ is obtained from $T$ by keeping only the transitions of the form $\delta(s,a)$ with $a\in\{0,1,\#\}$ and of the form $\delta(f,a)$ for $a\in\Sigma$, and rerouting all other transitions to state $\Call_2$.

\textit{Claim:}  $\mathcal A_2$ satisfies \eqref{it:compareA2}:\\
 Let $w\in\Sigma^*$. If $w\in\{0,1\}^*$, then the strategy $\sigma_{\mathcal A_2}$ plays only Read moves and the $\mathcal A_2$-state is $s$ all the time, yielding a play that is won by Romeo since, by our assumption, $s\notin F$. This shows that $W(\sigma_{\mathcal A_2})\subseteq\Sigma^*\setminus\{0,1\}^*$.

On the other hand, if $w\in\Sigma^*\setminus\{0,1\}^*$, let $x\in\Sigma\setminus\{0,1\}$ be the first symbol in $w$ that is not in $\{0,1\}$. If $x=\#$, the strategy $\sigma_{\mathcal A_2}$ reads $\#$ and the new $T$-state is $f$.  If $x\neq\#$ then $\sigma_{\mathcal A_2}$ calls $x$, it gets replaced by $\#$, and again $\sigma_{\mathcal A_2}$ reads $\#$ and reaches $f$. The state $f$ cannot be left any more, and since it is accepting, Juliet wins. Thus, $w\in W(\sigma_{\mathcal A_2})$, which completes the proof of the claim.

The proof of Proposition~\ref{prop:compare-lower} is now completed by the observation that $G=(\Sigma,R,T)$, $\mathcal A_1$ and $\mathcal A_2$ can clearly be constructed in polynomial time. 
\end{proof}
}% \full

Theorem~\ref{thm:complexity} assumes (as per our standard definition) that target language automata are DFAs. In \cite{Coester15}, also the case of NFAs was studied. The following complexities could be obtained.
\begin{itemize}
\item \textsc{IsWinning} is $\PSPACE$-complete. \cite[Theorem 6.2]{Coester15}
\item \textsc{ExistsWinning}(\SReg) is $\PSPACE$-hard and in $\NEXPTIME$. \cite[Theorem 6.7]{Coester15}
\item \textsc{IsDominated} is in $\EXPSPACE$. \cite[Corollary 6.10]{Coester15}
\end{itemize}

The problem \textsc{ExistsWinning} can also be studied for more general classes of strategies than strongly regular ones. We do not know in general whether this problem is decidable, even for regular (but not strongly regular) strategies. However, \cite[Theorem 6.8]{Coester15} shows that it is decidable for several restricted games, including
\begin{itemize}
\item prefix-free games,
\item non-recursive games with finite replacement sets,
\item non-recursive games with finite target languages.
\end{itemize}
\section{Conclusion and open questions}\label{cha:summary}
Whereas we have no positive results for general games, the results for restricted games are somewhat encouraging.  Indeed, prefix-free replacement languages seem to be a realistic solution in practice, because they can be achieved easily and with almost no overhead by suffixing replacement words with an end-of-file symbol. Another setting in which many positive properties hold is if the target language is finite. While finiteness of the target language may seem like a very strong restriction, one particular instance of it is if there is a single target document that has to be reached~\cite{MuschollSS06}. Restrictions that bound the number of recursive replacements also seem plausible in practise.

We have seen that in several cases it is possible to use a one-pass strategy that is based on an automaton and there is hope that this is possible in general (for undominated strategies). However, we can rule out that one can restrict oneself to strongly regular strategies without losing strategic power. We have also seen that it is not generally possible to \emph{find} good strategies efficiently.

Many questions remain open for further study. First of all, it remains unclear whether undominated strategies always exist. It is actually conceivable that each game has the bounded-depth property and therefore also has undominated strategies. It is also not clear whether in the bounded-depth property the bounding sequence can always be replaced by some constant. 
Another question is whether every game with an undominated strategy also has a regular one. This is open for general games and for games with a forgetful undominated strategy. 
Some open problems regarding complexity were mentioned at the end of Section~\ref{cha:complexity}.

\bibliographystyle{plainurl}
\bibliography{literature}

\end{document}